\documentclass{lmcs}
\pdfoutput=1

\usepackage{lastpage}
\lmcsdoi{15}{4}{12}
\lmcsheading{}{\pageref{LastPage}}{}{}%
{Dec.~07,~2018}{Dec.~04,~2019}{}

\usepackage{amsmath,amssymb}
\usepackage{xspace}
\usepackage[utf8]{inputenc}
\usepackage[textwidth=20mm]{todonotes}
\usepackage{microtype}
\usepackage{hyperref}

\usepackage{tikz}
\usetikzlibrary{automata,shapes.multipart,positioning,calc}

\newtheorem{theorem}[thm]{Theorem}
\newtheorem{lemma}[thm]{Lemma}
\newtheorem{corollary}[thm]{Corollary}
\newtheorem{definition}[thm]{Definition}
\newtheorem{proposition}[thm]{Proposition}
\newtheorem{claim}[thm]{Claim}

\newcommand{\FPT}{\textsf{FPT}\xspace}
\newcommand{\W}[1]{\textsf{W[#1]}\xspace}
\newcommand{\Wh}[1]{\textsf{W[#1]\mbox{-}\nobreak\hspace{0pt}hard}\xspace}
\newcommand{\Whss}[1]{\textsf{W[#1]\mbox{-}\nobreak\hspace{0pt}hardness}\xspace}
\newcommand{\XP}{\textsf{XP}\xspace}

\newcommand{\NP}{\textsf{NP}\xspace}
\newcommand{\Pee}{\textsf{P}\xspace}

\newcommand{\RR}{\mathbb{R}}
\newcommand{\ZZ}{\mathbb{Z}}
\newcommand{\NN}{\mathbb{N}}
\newcommand\AAA{\mathcal{A}}
\newcommand\BB{\mathcal{B}}
\newcommand\CC{\mathcal{C}}
\newcommand\DD{\mathcal{D}}
\newcommand\FF{\mathcal{F}}
\newcommand\HH{\mathcal{H}}
\newcommand\TT{\mathcal{T}}
\newcommand\SSS{\mathcal{S}}
\newcommand{\Oh}{\mathcal{O}}

\newcommand{\vc}[1]{\mathbf{#1}}
\newcommand{\vcl}{\mbox{\boldmath$\ell$}}

\newcommand{\LCCSubset}{{\sc LCC Subset}\xspace}
\newcommand{\kMCC}{$k$\mbox{-}\nobreak\hspace{0pt}{\sc Multicolored Clique}\xspace}

\newcommand{\probl}[3]{
\smallskip
\begin{center}
\fbox{
\begin{minipage}[c]{.92\textwidth}
#1\\
\textbf{Input:} #2\\
\textbf{Task:} #3
\end{minipage}}~\\
\end{center}
\smallskip
}

\newcommand{\Feas}{{\rm Feas}}
\newcommand{\MSO}{{\sf MSO}\xspace}
\newcommand{\cMSO}{{\sf Card}\MSO}
\newcommand{\fairMSO}{{\sf fair}\MSO}
\newcommand{\MSOLCC}{\MSO\unskip\mbox{-}\nobreak\hspace{0pt}{\sf LCC}\xspace}

\newcommand{\MSOG}{$\MSO\unskip^{\sf G}$\xspace}
\newcommand{\MSOGLin}{$\MSO\unskip^{\sf G}_{\sf lin}$\xspace}
\newcommand{\MSOL}{$\MSO\unskip^{\sf L}$\xspace}
\newcommand{\MSOLLin}{$\MSO\unskip^{\sf L}_{\sf lin}$\xspace}
\newcommand{\MSOGLLin}{$\MSO\unskip^{\sf GL}_{\sf lin}$\xspace}
\newcommand{\MSOGL}{$\MSO\unskip^{\sf GL}$\xspace}

\newcommand{\FO}{{\sf FO}\xspace}
\newcommand{\MSOJ}{{\sf MSO$_{1}$}\xspace}
\newcommand{\MSOD}{{\sf MSO$_{2}$}\xspace}
\newcommand{\A}{\ensuremath{\mathcal{A}}}
\newcommand{\E}{\ensuremath{\mathsf{E}}}
\newcommand{\NE}{\ensuremath{\mathsf{NE}}}

\DeclareMathOperator{\Otw}{tw}
\DeclareMathOperator{\Ond}{nd}
\newcommand{\tw}{\ensuremath{\Otw}}

\newcommand{\Gvc}{\ensuremath{\Ovc}}
\renewcommand{\nd}{\ensuremath{\Ond}}
\DeclareMathOperator{\Ovc}{vc}

\newcommand{\lb}{\mathrm{lb}}
\newcommand{\ub}{\mathrm{ub}}

\newcommand{\many}{\mathord{\uparrow}}

\newcommand{\heading}[1]{\subsubsection*{#1.}}
\begin{document}

\title[Simplified Algorithmic Metatheorems Beyond MSO]{Simplified Algorithmic Metatheorems Beyond MSO:\\ Treewidth and Neighborhood Diversity}
\thanks{}

\author[D.~Knop]{Dušan Knop}	
\address{
Department of Theoretical Computer Science, Faculty of Information Technology, \newline Czech Technical University in Prague, Prague, Czech Republic}
\email{knop@kam.mff.cuni.cz}  
\thanks{}	

\author[M.~Koutecký]{Martin Koutecký}	
\address{Faculty of IE\&M, Technion -- Israel Institute of Technology, Haifa, Israel
\newline
Computer Science Institute, Charles Univesity, Prague, Czech Republic}
\email{koutecky@iuuk.mff.cuni.cz}  
\thanks{}

\author[T.~Masařík]{Tomáš Masařík}	
\address{Department of Applied Mathematics, Charles University, Prague, Czech Republic
\newline
Faculty of Mathematics, Informatics and Mechanics, University of Warsaw, Poland}
\email{masarik@kam.mff.cuni.cz}  
\thanks{}	

\author[T.~Toufar]{Tomáš Toufar}	
\address{Computer Science Institute, Charles University, Prague, Czech Republic}
\email{toufi@iuuk.mff.cuni.cz}  


\subjclass{Theory of computation $\rightarrow$ Parameterized complexity and exact algorithms; Theory of computation $\rightarrow$ Logic; Theory of computation $\rightarrow$ Graph algorithms analysis}

\titlecomment{Preliminary version appeared in the proceedings of the WG 2017 conference~\cite{KnopKMT17}.}

\begin{abstract}
This paper settles the computational complexity of model checking of several extensions of the monadic second order (\MSO) logic on two classes of graphs: graphs of bounded treewidth and graphs of bounded neighborhood diversity.

A classical theorem of Courcelle states that any graph property definable in \MSO is decidable in linear time on graphs of bounded treewidth.
Algorithmic metatheorems like Courcelle's serve to generalize known positive results on various graph classes.
We explore and extend three previously studied \MSO extensions: global and local cardinality constraints (\cMSO and \MSOLCC) and optimizing the fair objective function (\fairMSO).

First, we show how these extensions of \MSO relate to each other in their expressive power.
Furthermore, we highlight a certain ``linearity'' of some of the newly introduced extensions which turns out to play an important role.
Second, we provide parameterized algorithms for the aforementioned structural parameters.
On the side of neighborhood diversity, we show that combining the linear variants of local and global cardinality constraints is possible while keeping the linear (\FPT) runtime but removing linearity of either makes this impossible assuming $\FPT \neq \W{1}$.
Moreover, we provide a polynomial time (\XP) algorithm for the most powerful of studied extensions, i.e.\ the combination of global and local constraints.
Furthermore, we show a polynomial time (\XP) algorithm on graphs of bounded treewidth for the same extension.
In addition, we propose a general procedure for deriving \XP algorithms on graphs on bounded treewidth using Constraint Satisfaction Problems (CSPs).
This shows an alternative approach to standard dynamic programming formulations.
\end{abstract}
\maketitle

\section{Introduction}\label{sec:introduction}

It has been known since the '80s that various NP-hard problems are solvable in polynomial time by dynamic programming on trees and ``tree-like'' graphs. This was famously captured by Courcelle~\cite{Courcelle:90} in his theorem stating that any property definable in Monadic Second Order (\MSO) logic is decidable in linear time on graphs of bounded treewidth.
Subsequently, extensions to stronger logics and optimization versions were devised~\cite{ALS:91,CM:93} while still keeping linear runtime.

However, several interesting problems do not admit an \MSO description and are unlikely to be solvable in linear time on graphs of bounded treewidth due to hardness results. In the language of parameterized complexity, Courcelle's theorem runs in \textit{fixed-parameter tractable} (\FPT) time, that is, in time $f(||\varphi||, \tau)n^{\Oh(1)}$, where $n$ is the number of vertices of the input graph, $\tau$ its treewidth, $\varphi$ is an \MSO formula, and $f$ is a~computable function. On the other hand, unless \FPT{}=\W{1}, the ``hard'' (specifically, \Wh{1}) problems have algorithms running at best in \XP time  $n^{g(||\varphi||, \tau)}$, for some computable function $g \in \omega(1)$.
This led to an examination of extensions of \MSO which allow greater expressive power.

Another research direction was to improve the computational complexity of Courcelle's theorem, since the function $f$ grows as an exponential tower in the number of quantifier alternations of the \MSO formula. However, Frick and Grohe~\cite{FrickG:04} proved that this is unavoidable unless $\Pee=\NP$ which raises a question: is there a (simpler) graph class where \MSO model checking can be done in single-exponential (i.e., $2^{k^{\Oh(1)}}$) time? This was answered in the affirmative by Lampis~\cite{Lampis:12}, who introduced graphs of bounded neighborhood diversity.
The classes of bounded treewidth and bounded neighborhood diversity are incomparable: for example, paths have unbounded neighborhood diversity but bounded treewidth, and vice versa for cliques.
Bounded treewidth has become a standard parameter with many practical applications (cf. a survey~\cite{Bodlaender:06}); bounded neighborhood diversity is of theoretical interest~\cite{AlvesDFKSS16,AravindKKL:2017,BonnetS:2016,FialaGKKK:16,Ganian:12,Gargano:2015,MasarikT:16} because it can be viewed as representing the simplest of dense graphs.

Courcelle's theorem proliferated into many fields. Originating among automata theorists, it has since been reinterpreted in terms of finite model theory~\cite{Libkin:04}, database programming~\cite{GPW:07}, game theory~\cite{KLR:11} and linear programming~\cite{KolmanKT:2015}.

\subsection{Related work}
For a recent survey of algorithmic metatheorems see e.g. Grohe et al.~\cite{grohe2011methods}.

\noindent\textbf{Objective functions.} A linear optimization version of Courcelle's theorem was given by Arnborg, Lagergren and Seese~\cite{ALS:91}. An extension to further objectives was given by Courcelle and Mosbah~\cite{CM:93}. Kolman, Lidický and Sereni~\cite{KolmanLS:09} introduce \MSO with a fair objective function (\fairMSO) which, for a given \MSO formula $\varphi(F)$ with a free edge set variable $F$, minimizes the maximum degree in the subgraph given by $F$, and present an \XP algorithm. This is justified by the problem being \Wh{1}, as was later shown by Masařík and Toufar~\cite{MasarikT:16}, who additionally gave an \FPT algorithm on graphs of bounded neighborhood diversity for \MSOJ and an \FPT algorithm on graphs of bounded vertex cover for \MSOD. Those results were extended to graphs of bounded twin cover by Knop, Masařík and Toufar~\cite{KMT17}.

\noindent\textbf{Extended logics.} Along with \MSO, Courcelle also considered \textit{counting \MSO} (\textsf{cMSO}) where predicates of the form ``$|X| \equiv p \mod q$'' are allowed, with the largest modulus $q$ constant. Szeider~\cite{Szeider:11} introduced \MSO with \textit{local cardinality constraints} (\MSOLCC) and gave an \XP algorithm deciding it on graphs of bounded treewidth. \MSOLCC can express various problems, such as \textsc{General Factor}, \textsc{Equitable $r$-Coloring} or \textsc{Minimum Maximum Outdegree}, which are known to be \Wh{1} on graphs of bounded treewidth. Ganian and Obdržálek~\cite{GO:13} study \textsf{CardMSO}, which is incomparable with \MSOLCC in its expressive power; they give an \FPT algorithm on graphs of bounded neighborhood diversity.

\subsection{Our contribution}

The contribution of the~paper is twofold.
First, we survey and enrich the so far studied extensions of \MSO logic -- \fairMSO, \cMSO, and \MSOLCC.
We do this in Section~\ref{subsec:mso_and_its_extensions}.
Second, we study the parameterized complexity of the associated model checking problem for various combinations of these \MSO extensions.
We completely settle the parameterized complexity landscape for the model checking problems with respect to the parameters treewidth and neighborhood diversity; for an overview of the complexity landscape refer to Figures~\ref{fig:extensions} and~\ref{fig:nd_complexity}.
We postpone formal definitions of logic extensions and corresponding model checking to Subsection~\ref{subsec:mso_and_its_extensions}.

While both \MSOLCC and \cMSO express certain \textit{cardinality} constraints, the constraints of \cMSO are inherently \textit{global} and \textit{linear}, yet the constraints of \MSOLCC are \textit{local} and \textit{non-linear}.
This leads us to introduce two more fragments and to rename the aforementioned ones: \cMSO becomes \MSOGLin, \MSOLCC becomes \MSOL and we additionally have \MSOG, \MSOLLin, \MSOGLLin and \MSOGL.
By this we give a complete landscape for all possible combinations of types of constraints: global/local and linear/non-linear.

In the following, we do not differentiate between the logics \MSOJ (allowing quantification over vertex sets) and \MSOD (additionally allowing quantification over edge sets); see a detailed explanation in Subsection~\ref{subsec:mso_and_its_extensions}.
For now, it suffices to say that our positive result for graphs of bounded treewidth holds for the appropriate extension of \MSOD, while all remaining results (positive for graphs of bounded neighborhood diversity and negative for graphs of bounded vertex cover number) hold for the appropriate extensions of \MSOJ.

For graphs of bounded treewidth we give an \XP algorithm for the logic \MSOGL, which is a composition of \MSOG and \MSOL and thus represents the most expressive fragment under our consideration.
\begin{theorem}\label{thm:HYDEisXPwrtTW}
\MSOGL model checking is \XP parameterized by $\tw(G)$ and $||\varphi||$.
\end{theorem}

This result is also significant in its proof technique.
We connect a recent result of Kolman, Koutecký and Tiwary~\cite{KolmanKT:2015} about the polytope of satisfying assignments with an old result of Freuder~\cite{Freuder:90} about the solvability of the constraint satisfaction problem (CSP) of bounded treewidth.
This allows us to formulate the proof of Theorem~\ref{thm:HYDEisXPwrtTW} essentially as providing a CSP instance with certain properties.
We also briefly discuss which problems can be modeled using \MSOGL, deriving the following as a consequence:

\begin{corollary}\label{thm:HYDE_applications}
Let $G$ be a graph of treewidth $\tau$ and with $n = |V(G)|$.

The following problems have algorithms with runtime $n^{f(\tau)}$:
\textsc{General Factor}, \textsc{Minimum Maximum Outdegree},
\textsc{Capacitated Dominating Set}, \textsc{Capacitated Vertex Cover}, \textsc{Vector Dominating Set}, \textsc{Generalized Domination}.

The following problems have algorithms with runtime  $n^{f(\tau + k)}$:
\begin{itemize}
\item \textsc{Equitable $k$-Coloring}, \textsc{Equitable connected $k$-Partition}, \textsc{$k$-Balanced Partitioning},
\item \textsc{Graph Motif}, where $k$ is the number of colors.
\end{itemize}
\end{corollary}

Theorem~\ref{thm:HYDEisXPwrtTW} is complemented from the negative side by the following hardness result.
\begin{theorem}\label{thm:MSOGisWHwrtVC}
\MSOG model checking is \Wh{1} parameterized by $\Gvc(G)$ and $||\varphi||$, where $\Gvc(G)$ is the vertex cover number of the input graph $G$.
\end{theorem}

For graphs of bounded neighborhood diversity we give two positive results; the logic \MSOGLLin is a composition of \MSOLLin and \MSOGLin.
\begin{theorem}\label{thm:JECKYLisFPTwrtND}
\MSOGLLin model checking is \FPT parameterized by $\nd(G)$ and $||\varphi||$.
\end{theorem}
\begin{theorem}\label{thm:HYDEisXPwrtND}
\MSOGL model checking is \XP parameterized by $\nd(G)$ and $||\varphi||$.
\end{theorem}

We complement the above results with another hardness result.
\begin{theorem}\label{thm:MSOLisWHwrtND}
\MSOL model checking is \Wh{1} parameterized by $\Gvc{(G)}$ and $||\varphi||$.
\end{theorem}

\begin{figure}
\begin{minipage}{.25\textwidth}

\usetikzlibrary{positioning,calc}

\begin{tikzpicture}[node distance=1.8cm]
\tikzstyle{sipka}=[->,thick]

\node (MSO) {\MSO};

\node[above left of=MSO] (MSOLLin) {\MSOLLin};
\node[above right of=MSO] (MSOGLin) {\MSOGLin};
\node[above left of=MSOGLin] (MSOGLLin) {\MSOGLLin};

\node[above left of=MSOGLLin] (MSOL) {\MSOL};
\node[above right of=MSOGLLin] (MSOG) {\MSOG};
\node[above left of=MSOG] (MSOGL) {\MSOGL};

\draw[sipka] (MSO) to (MSOGLin);
\draw[sipka] (MSO) to (MSOLLin);
\draw[sipka] (MSOLLin) to (MSOGLLin);
\draw[sipka] (MSOGLin) to (MSOGLLin);
\draw[sipka] (MSOLLin) to (MSOL);
\draw[sipka] (MSOGLLin) to (MSOGL);
\draw[sipka] (MSOGLin) to (MSOG);
\draw[sipka] (MSOG) to (MSOGL);
\draw[sipka] (MSOL) to (MSOGL);

\draw[ultra thick, green, dashed] ($ (MSOLLin.south west) $) to[out=-20,in=200] ($ (MSOGLin.south east) $);
\draw[ultra thick, dotted, orange] ($ (MSOL.south west) $) to[out=-20,in=200] ($ (MSOG.south east) $);
\draw[ultra thick, dash dot, red] ($ (MSOGL.north west) - (1,0)$) to[out=20,in=160] ($ (MSOGL.north east) + (1,0) $);
\end{tikzpicture}
\end{minipage}
\begin{minipage}{.62\textwidth}
\caption{\label{fig:extensions}%
{\small\textbf{\MSO extensions.}
A partial order of \MSO extensions considered.
An arrow denotes generalizations; e.g., \MSOLLin generalizes \MSO and is generalized by \MSOGLLin.
Green (dashed) line separates logics whose model checking is \FPT parameterized by $\tw(G)$ (Courcelle~\cite{Courcelle:90}) from those whose model checking is \Wh{1} (both \MSOLLin and \MSOGLin capture the \Wh{1} \textsc{Equitable $r$-Coloring} problem).
Orange (dotted) line separates logics whose model checking is \FPT parameterized by $\nd(G)$ (Theorem~\ref{thm:JECKYLisFPTwrtND}) from those whose model checking is \Wh{1} (Theorems~\ref{thm:MSOGisWHwrtVC} and~\ref{thm:MSOLisWHwrtND}).
The model checking of all logics below the red (dashed-dotted) line is \XP parameterized by both $\tw(G)$ (Theorem~\ref{thm:HYDEisXPwrtTW}) and $\nd(G)$ (Theorem~\ref{thm:HYDEisXPwrtND}).}
}
\end{minipage}
\end{figure}

Interestingly, our finding about hardness being caused by nonlinearity carries over to a generalization of the \textsc{Set Cover} problem.

\probl{\textsc{Multidemand Set Multicover}}
{Universe $U = [k]$, set of multidemands $D_1, \dots, D_k \subseteq [n]$, a covering system $\mathcal{F} = \{F_1, \dots, F_f\} \subseteq 2^U$, an integer $r \in \NN$}
{Find integer multiplicities $m_1 + \dots + m_f = r$ such that for all $i\in[k]$,
$(\sum_{j:u_i \in F_j} m_j) \in D_i$.}

We show that the (non)linearity of the multidemands is crucial.
On one hand, we have the following hardness.

\begin{theorem} \label{thm:msc_wh}
\textsc{Multidemand Set Multicover} is \Wh{1} parameterized by~$k$, already when $|\FF| = k$.
\end{theorem}

On the other hand, consider the \textsc{Weighted Set Multicover} problem.
It is a weighted variant of \textsc{Multidemand Set Multicover} where each multidemand $D_i$ is the interval $[0, \delta_i]$ for some $\delta_i \in \NN$.
Recently, Bredereck et al.~\cite{BredereckEtAl2015} showed that this variant is fixed-parameter tractable when parameterized by $k$, which has applications in computational social choice and elsewhere.
Moreover, it is not difficult to see that even the more general case when each $D_i$ is a discrete interval is in \FPT, by a small modification of the approach of Bredereck et al.

\begin{table}[bt]
  \begin{center}

\begin{tikzpicture}
  \tikzstyle{policko}=[rectangle, minimum height=.7cm, minimum width=3.0cm, inner sep=2pt]
  \tikzstyle{polickoG}=[policko, fill=green!30]
  \tikzstyle{polickoR}=[policko, fill=orange!90]
  \tikzstyle{sipka}=[ultra thick, >=stealth]

  \node[policko] (N) {nd, vc};
  \node[policko, below=0cm of N] (noCard) {\MSO};
  \node[policko, below=0cm of noCard] (card) {\MSOGLin};
  \node[policko, below=0cm of card] (GCC) {\MSOG};

  \node[policko, right=0cm of N] (noFair) {$\MSO$};
  \node[policko, right=0cm of noFair] (Fair) {\fairMSO};
  \node[policko, right=0cm of Fair] (LLCC) {\MSOLLin};
  \node[policko, right=0cm of LLCC] (LCC) {\MSOL};

  \node[polickoG, right=.0cm of noCard] (ncE) {\FPT~\cite{Lampis:12}};
  \node[polickoG, right=0cm of ncE] (ncF) {\FPT~\cite{MasarikT:16}};
  \node[polickoG, right=0cm of ncF] (ncL) {};
  \node[polickoR, right=0cm of ncL] (ncG) {\textsf{W[1]-h}, Thm~\ref{thm:MSOLisWHwrtND}};

  \node[polickoG, right=0cm of card] (cE) {\FPT~\cite{GO:13}};
  \node[polickoG, right=0cm of cE] (cF) {};
  \node[polickoG, right=0cm of cF] (cL) {\FPT, Thm~\ref{thm:JECKYLisFPTwrtND}};
  \node[polickoR, right=0cm of cL] (cG) {};

  \node[polickoR, right=0cm of GCC] (gE){\textsf{W[1]-h}, Thm~\ref{thm:MSOGisWHwrtVC}};
  \node[polickoR, right=0cm of gE] (gF) {};
  \node[polickoR, right=0cm of gF] (gL) {};
  \node[polickoR, right=0cm of gL] (gG) {\XP, Thm~\ref{thm:HYDEisXPwrtND}};

  \draw[ultra thick] (N.south west) -- (LCC.south east);
  \draw[ultra thick] (N.north east) -- (GCC.south east);

  \draw (noCard.south west) -- (ncG.south east);
  \draw (card.south west) -- (cG.south east);
  \draw (noFair.north east) -- (gE.south east);
  \draw (Fair.north east) -- (gF.south east);
  \draw (LLCC.north east) -- (gL.south east);

\end{tikzpicture}
\begin{tikzpicture}
  \tikzstyle{policko}=[rectangle, minimum height=.7cm, minimum width=3.0cm, inner sep=2pt]
  \tikzstyle{polickoG}=[policko, fill=green!30]
  \tikzstyle{polickoR}=[policko, fill=orange!90]
  \tikzstyle{sipka}=[ultra thick, >=stealth]

  \node[policko] (N) {tw};
  \node[policko, below=0cm of N] (noCard) {\MSO};
  \node[policko, below=0cm of noCard] (card) {\MSOGLin};
  \node[policko, below=0cm of card] (GCC) {\MSOG};

  \node[policko, right=0cm of N] (noFair) {$\MSO$};
  \node[policko, right=0cm of noFair] (Fair) {\fairMSO};
  \node[policko, right=0cm of Fair] (LLCC) {\MSOLLin};
  \node[policko, right=0cm of LLCC] (LCC) {\MSOL};

  \node[polickoG, right=.0cm of noCard] (ncE) {\FPT~\cite{Courcelle:90}};
  \node[polickoR, right=0cm of ncE] (ncF) {\Wh{1}~\cite{MasarikT:16}};
  \node[polickoR, right=0cm of ncF] (ncL) {};
  \node[polickoR, right=0cm of ncL] (ncG) {\XP~\cite{Szeider:11}};

  \node[polickoR, right=0cm of card] (cE) {\Wh{1}~\cite{GO:13}};
  \node[polickoR, right=0cm of cE] (cF) {};
  \node[polickoR, right=0cm of cF] (cL) {};
  \node[polickoR, right=0cm of cL] (cG) {};

  \node[polickoR, right=0cm of GCC] (gE) {};
  \node[polickoR, right=0cm of gE] (gF) {};
  \node[polickoR, right=0cm of gF] (gL) {};
  \node[polickoR, right=0cm of gL] (gG) {\XP, Thm~\ref{thm:HYDEisXPwrtTW}};

  \draw[ultra thick] (N.south west) -- (LCC.south east);
  \draw[ultra thick] (N.north east) -- (GCC.south east);

  \draw (noCard.south west) -- (ncG.south east);
  \draw (card.south west) -- (cG.south east);
  \draw (noFair.north east) -- (gE.south east);
  \draw (Fair.north east) -- (gF.south east);
  \draw (LLCC.north east) -- (gL.south east);

\end{tikzpicture}
	\caption{Complexity of various logic fragments generalizing \MSO parameterized by $||\varphi||$ and in addition by vertex cover ($\Gvc$), neighborhood diversity ($\nd$), or treewidth ($\tw$).
  Positive results (\FPT, \XP) spread to the left and up. \Whss{1} spreads to the right and down. Green background (lighter gray in b-w print) stands for \FPT fragments, while orange (darker gray) stands for \Wh{1}. A cell represents a union of studied fragments, i.e.,\ the cell indexed by \MSOG and \MSOL corresponds to the \MSOGL fragment. }\label{fig:nd_complexity}
  \end{center}
\end{table}

\section{Preliminaries}\label{sec:preliminaries}
For two integers $a,b$ we define a set $[a,b] = \{x\in\ZZ\mid a\le x \le b\}$; we use $[b]$ to denote the set $[1,b]$.
We write vectors in bold font, i.e., $\vc x \in \RR^n$, and denote by $x_i$, $i \in [n]$, its $i$-th coordinate.

For a vertex $v\in V$ of a~graph $G=(V,E)$, we denote by $N_G(v)$ the set of neighbors of $v$ in $G$, that is, $N_G(v)=\left\{u\in V\mid \{u,v\}\in E\right\}$; the subscript $G$ is omitted when clear from the context.
For a rooted tree $T$, $N^{\downarrow}_T(v)$ denotes the \emph{down-neighborhood} of $v$, i.e., the set of children of $v$.
For a~graph $G=(V,E)$ a~set $U\subseteq V$ is a \emph{vertex cover} of $G$ if for every edge $e\in E$ it holds that $e\cap U\neq\emptyset$.
The size of a minimum vertex cover of graph $G$ is denoted by $\Gvc(G)$.
For more graph theory notation cf.\ the book of Matoušek and Nešetřil~\cite{kapitoly}.

\subsection{\MSO and its Extensions}
\label{subsec:mso_and_its_extensions}
Let us shortly introduce \MSO over graphs.
In {\em first-order logic} (\FO) we have variables for the elements
($x,y,\ldots$), equality for variables, quantifiers $\forall,\exists$
ranging over elements, and the standard Boolean connectives $\neg, \land, \lor, \implies$.
The \emph{monadic second order logic} (\MSO) extends first order logic using so called monadic variables.
Graph \MSO has the binary relational symbol 
$E$ encoding edges, and traditionally comes in two
flavors, \MSOJ and \MSOD, differing by the objects we are allowed to quantify over:
in \MSOJ these are the vertices and vertex sets,
while in \MSOD we can additionally quantify over edges and edge sets.
For example, the $3$-colorability property can be expressed in \MSOJ as follows:
\begin{eqnarray*}
 \exists X_1,X_2,X_3 &&\left[\,
 	\forall x \, (x\in X_1\vee x\in X_2\vee x\in X_3) \wedge \right.
\\ &&\bigwedge\nolimits_{i=1,2,3} \left. \!\!
   \forall x,y
   \left(x\not\in X_i\vee y\not\in X_i\vee \{x,y\}\not\in E\right)
 \,\right]
\end{eqnarray*}
For a formula $\varphi$, we denote by $||\varphi||$ the \emph{size} of $\varphi$, that is, the number of symbols in a chosen encoding of $\varphi$.

\subsubsection*{Regarding \MSOJ and \MSOD}\label{par:msotw}
We use standard notation for relational structures (cf. Libkin~\cite{Libkin:04}):
a \emph{vocabulary $\sigma$} is a collection of constant and predicate symbols, and a \emph{$\sigma$-structure} is a~tuple $\A=\bigl(A,\{c_i\},\{P_i\}\bigr)$, where the set $A$ is the \emph{universe}, $\{c_i\}$ are \emph{constant symbols} ($c_i\in A$) and $\{P_i\}$ are finitely many \emph{relations} (or \emph{predicates}), each of arity $r_i$ ($P_i\subseteq A^{r_i}$).
(More precisely, $c_i$ and $P_i$ are the realizations of the symbols in the vocabulary $\sigma$, but we identify a symbol and its realization to avoid notational overload.)
Then, a graph $G=(V,E)$ is typically modeled as a $\sigma_1$-structure $\bigl(V,\emptyset,\{E\}\bigl)$ with $\sigma_1$ containing one binary relation $E$.
We can view $G$ as a $\sigma_2$-structure ${\A=\bigl(V\cup E,\emptyset,\{I,L_V,L_E\}\bigr)}$, where $\sigma_2$ contains a binary relation $I$ representing incidence in $G$ and two unary predicates $L_V, L_E$ distinguishing vertices and edges, respectively, with realizations defined as $I=\{\{v,e\}\mid {v\in e}, {e\in E}\}$, $L_V = V$ and $L_E = E$.
Given a vocabulary $\sigma$, we define the logic $\MSO[\sigma]$ inductively as usual by letting its terms be the variables and constant symbols, its atomic formulae (atoms) be either $t_1 = t_2$ for two terms $t_1, t_2$, or $P(t_1, \dots, t_k)$ for any $k$-ary predicate $P$ and terms $t_1, \dots, t_k$, and its formulae be quantified boolean combinations of atoms.
Thus, $\MSO[\sigma_1]$ is the logic with a binary predicate $E$, and $\MSO[\sigma_2]$ is the logic with two unary predicates $L_V, L_E$ and one binary predicate $I$.
For short, we denote \MSOJ $= \MSO[\sigma_1]$ and \MSOD $= \MSO[\sigma_2]$.
The treewidth of a relational structure $\A$ is the treewidth of its \emph{Gaifman graph} $G(\A) = (V^{\A},E^{\A})$ where $V^{\A} = A$ and $E^{\A} = \{\{u,v\} \mid \exists P_i \exists \vc r \in P_i: u,v \in \vc r \}$.
It is known that the treewidth of a graph does not change when viewed as a $\sigma_1$-structure, and increases by at most $1$ when viewed as a $\sigma_2$-structure~\cite{Kolaitis:00}.
Even though \MSOD is strictly stronger than \MSOJ (hamiltonicity is expressible in \MSOD but not in \MSOJ~\cite{Libkin:04}), on graphs with bounded treewidth their power is equal~\cite{Courcelle:1994,Courcelle:2003}.

%

For this reason, when considering graphs of bounded treewidth, we will focus on \MSOD, and all the extended logics such as \MSOGLLin etc. will be extensions of \MSOD.
On the other hand, on graphs of bounded neighborhood diversity, \MSOD is strictly more powerful than \MSOJ. However, because model checking of an \MSOD formula already over a \emph{clique} is not even in \XP unless $\E=\NE$~\cite{Courcelle:00,Lampis:13}, when dealing with graphs of bounded neighborhood diversity we will always refer to \MSOJ and \MSOGLLin will stand for an extension of \MSOJ.

\medskip

We consider two orthogonal ways to extend \MSO logic. In what follows $\varphi$ is a formula with $\ell$ free set-variables.

\subsubsection*{Global cardinality constraints.} We introduce a new type of atomic formulae called \emph{global cardinality constraints} (\emph{global constraints} for short). An \MSO formula with $c$ global cardinality constraints contains $\ell$-ary predicates $R_1, \ldots, R_c$ where each predicate takes as argument only the free variables of $\varphi$.
The input to the model checking problem is a graph $G = (V,E)$ on $n$ vertices and an implicit representation (see below) of a tuple $(R_1^G, \ldots, R_c^G)$, where $R_i^G \subseteq [n]^\ell$.

To define the semantics of the extension, it is enough to define the truth of the newly introduced atomic formulae.
A formula $R_i(X_1, \ldots, X_\ell)$ is true under an assignment $\mu\colon\{X_1,\ldots, X_\ell\} \to 2^V$ if and only if ${(|\mu(X_1)|,\ldots,|\mu(X_\ell)|) \in R_i^G}$.
We allow a relation $R_i^G$ to be represented either as a linear constraint $a_1 |X_1| + \cdots + a_m |X_m| \leq b$, where $(a_1, \dots, a_m, b) \in \RR^{m+1}$, or, much more generally, by an arbitrary algorithm $\AAA(R_i^G)$ such that  $\AAA(R_i^G)$ decides whether $(|X_1|, \dots, |X_\ell|) \in R_i^G$ in time $n^{\Oh(1)}$ for any tuple in $[0,n]^\ell$.

For example, suppose we want to satisfy a formula $\varphi(X_1, X_2)$ with two sets for which $|X_1| \geq |X_2|^2$ holds.
Then, we solve the \MSOG model checking problem with a formula $\varphi' := \varphi \wedge [|X_1| \geq |X_2|^2]$, that is, we write the relation as a part of the formula, as this is a more convenient way to think of the problem.
However, formally the relation is a part of the input, represented by the obvious algorithm computing $|X_2|^2$, comparing it with $|X_1|$, and returning the result.

\subsubsection*{Local cardinality constraints.}
Local cardinality constraints are additional cardinality requirements such that every variable assignment has to satisfy the cardinality constraint for every vertex and for every free variable.
Specifically, we want to control the size of $\mu(X_i) \cap N(v)$ for every $v$; we define a shorthand  $S(v) = S \cap N(v)$ for a subset $S \subseteq V$ and vertex $v$.
\emph{Local cardinality constraints} for a graph $G = (V,E)$ on $n$ vertices and a formula $\varphi$ with $\ell$ free variables are mappings $\alpha_1, \ldots, \alpha_\ell$, where each $\alpha_i$ is a mapping from $V$ to $2^{[n]}$.

We say that an assignment $\mu$ \emph{obeys local cardinality constraints $\alpha_1,\ldots,\alpha_\ell$} if for every $i \in [\ell]$ and every $v \in V$ it holds that $|\mu(X_i)(v)| \in \alpha_i(v)$.

\medskip

The logic that incorporates both of these extensions is denoted \MSOGL.
Let $\varphi$ be an \MSOGL formula with $c$ global cardinality constraints.
Then the \MSOGL model checking problem has input:
\begin{itemize}
	\item graph $G = (V,E)$ on $n$ vertices,
	\item relations $R_1^G, \ldots, R_c^G \subseteq [n]^\ell$, and,
	\item mappings $\alpha_1, \ldots, \alpha_\ell$.
\end{itemize}
The task is to find an assignment $\mu$ that obeys local cardinality constraints and such that $\varphi$ is true under $\mu$ by the semantics defined above.

The \MSOGL logic is very powerful and, as we later show, it does not admit an \FPT model checking algorithm neither for the parameterization by neighborhood diversity, nor for the parameterization by treewidth.
It is therefore relevant to consider the following weakenings of the \MSOGL logic:

\begin{description}
\item[\MSOG] Only global cardinality constraints are allowed.
\item[\MSOL {\normalfont (originally \MSOLCC~\cite{Szeider:11})}] Only local cardinality constraints are allowed.
\item[\MSOGLin {\normalfont (originally \cMSO~\cite{Ganian:12})}] The cardinality constraints can only be linear; that is, we allow constraints in the form $[e_1 \geq e_2]$, where $e_i$ is a linear expression over $|X_1|, \ldots, |X_\ell|$.
\item[\MSOLLin] Only local cardinality constraints are allowed; furthermore every local cardinality constraint $\alpha_i$ must be of the form $\alpha_i(v) = [l_i^v,u_i^v]$, (i.e., an interval) where $l_i^v,u_i^v \in [n]$. Those constraints are referred to as \emph{linear local cardinality constraints}.
\item[\fairMSO] Further restriction of \MSOLLin; now we only allow $\alpha_i(v) = [u_i^v]$.
\item[\MSOGLLin] A combination of \MSOLLin and \MSOGLin; both local and global constraints are allowed, but only in their linear variants.
\end{description}
The model checking problem for the considered fragments is defined in a natural way analogously to \MSOGL model checking.

\subsubsection*{Pre-evaluations.}
Many techniques used for designing \MSO model checking algorithms fail when applied to \MSO extensions. A common workaround is first transforming the given \MSOGL formula into an \MSO formula by fixing the truth values of all global constraints to either \texttt{true} or \texttt{false}.
Once we determine which variable assignments satisfy the transformed \MSO formula, we can by other means (e.g.\ integer linear programming or constraint satisfaction) ensure that they obey the constraints imposed by fixing the values to \texttt{true} or \texttt{false}.
This approach was first used for \cMSO by Ganian and Obdržálek~\cite{GO:13}.
We formally describe this technique as \emph{pre-evaluations}:

\begin{definition}[Pre-evaluation]
\label{def:pre-evaluation}
Let $\varphi$ be an \MSOGL formula. Denote by $C(\varphi)$ the list of all global constraints. A mapping $\beta \colon C(\varphi) \to \{\mathtt{true}, \mathtt{false}\}$ is called a \emph{pre-evaluation function on $\varphi$}. The \MSO formula obtained by replacing each global constraint $c_i \in C(\varphi)$ by $\beta(c_i)$ is denoted by $\beta(\varphi)$ and is referred to as a \emph{pre-evaluation of $\varphi$}.
\end{definition}

\begin{definition}[Assignment and Pre-evaluation Compliance]
\label{def:complying_pre-evaluation}
A variable assignment $\mu$ of an \MSOGL formula $\varphi$ \emph{complies} with a pre-evaluation function $\beta$ if every global constraint $c_i \in C(\varphi)$ evaluates to $\beta(c_i)$ under the assignment $\mu$.
\end{definition}

\subsection{Treewidth and Neighborhood Diversity}

\heading{Treewidth}
For notions related to the treewidth of a graph and nice tree decompositions, in most cases we stick to the standard terminology as given by Kloks~\cite{Kloks:94}; the only deviation is in the leaf nodes of the nice tree decomposition where we assume that the bags are empty.

\begin{definition}[Tree decomposition, Treewidth]\label{def:tw}
A \emph{tree decomposition} of a graph $G=(V,E)$
is a pair $(T, \BB)$, where $T$ is a tree and $\BB$ is a mapping
$\BB: V(T) \rightarrow 2^V$ satisfying
\begin{itemize}
	\setlength\itemsep{0em}
	\item for any $\{u,v\} \in E$, there exists $a \in V(T)$ such that
	$u, v \in B(a)$,
	\item if $v \in B(a)$ and $v \in B(b)$, then $v \in B(c)$ for all
	$c$ on the $a$-$b$ path in $T$.
\end{itemize}
We call the vertices of the tree \emph{nodes} and the sets $B(a)$ we call \emph{bags}.

The {\em treewidth $tw((T, \BB))$ of a tree decomposition} $(T, \BB)$ is the size of the largest bag of $(T, \BB)$ minus one.
A graph $G$ has {\em treewidth} $\tau$ ($\tw(G)=\tau$) if it has a tree decomposition of treewidth $\tau$.
\end{definition}

Observe that for every graph $G$ we have $\tw(G) \le \Gvc(G)$, since it admits a tree decomposition with $T$ being a path $P_{|V(G)| - \Gvc(G)}$.
In this decomposition we have (in any order) bags of the form $U \cup \{ v \}$ for each vertex $v \in V \setminus U$, where $U$ is a vertex cover of $G$ with $|U| = \Gvc(G)$.
It is easy to verify that the constructed decomposition is a tree decomposition of $G$.
Thus, the class of graphs of bounded treewidth is more general than the class of graphs of bounded vertex cover.

\begin{definition}[Nice tree decomposition]\label{def:ntw}
A \emph{nice tree decomposition}
is a tree decomposition with $T$ rooted and binary, where the root is denoted $r$ and each node is
one of the following types:
\begin{itemize}
	\setlength\itemsep{0em}
	\item \emph{Leaf node}: a leaf $a$ of $T$ with $B(a) = \emptyset$.
	\item \emph{Introduce node}: an internal node $a$ of $T$ with one
	child $b$ for which $B(a) = B(b) \cup \{v\}$ for some $v \in
	B(a)$; for short we write $a = b*(v)$
	\item \emph{Forget node}: an internal node $a$ of $T$ with one child
	$b$ for which $B(a) = B(b) \setminus \{v\}$ for some $v \in B(b)$; for short $a = b\dagger(v)$
	\item \emph{Join node}: an internal node $a$ with two children $b$
	and $c$ with $B(a) = B(b) = B(c)$; for short $a = \Lambda(b,c)$.
\end{itemize}
\end{definition}
For a vertex $v\in V$, we denote by $top(v)$ the topmost
node of a nice tree decomposition that contains $v$ in its bag.
For any graph $G$ on $n$ vertices, a nice tree decomposition of $G$ with width $\tw(G)$ and
at most $8n$ nodes can be computed
in time $\Oh(n)$~\cite{Bodlaender:93,Kloks:94}.
(This is done by first using Bodlaender's algorithm~\cite{Bodlaender:93} to compute an optimal tree decomposition with at most $n$ nodes in linear time, and then transforming it into a nice decomposition using an algorithm of Kloks~\cite{Kloks:94} while keeping the number of nodes bounded by $8n$.)

Given a graph $G=(V,E)$ and a subset of vertices $V'=\{v_1, \dots, v_d\}\subseteq V$,
we denote by $G[V']$ the subgraph of $G$ induced by $V'$.
Given a tree decomposition $(T, \BB)$ and a node $a \in V(T)$,
we denote by $T_a$ the subtree of $T$ rooted in
$a$, and by $G_a$ the subgraph of $G$ induced by all vertices in
bags of $T_a$, that is, $G_a = G[\bigcup_{b \in V(T_a)} B(b)]$.

\heading{Neighborhood diversity}
We say that two (distinct) vertices $u,v$ are of the same {\em neighborhood type} if they share their respective neighborhoods, that is when ${N(u)\setminus\{v\} = N(v)\setminus\{u\}}.$
Let $G=(V,E)$ be a graph. We call a partition of vertices $\TT = \{T_1, \dots, T_\nu\}$ a \textit{neighborhood decomposition} if, for every $i \in [\nu]$, all vertices of $T_i$ are of one neighborhood type.

\begin{definition}[Neighborhood diversity~\cite{Lampis:12}]\label{def:nd}
A graph $G = (V,E)$ has {\em neighborhood diversity} $\nu$ ($\nd(G) = \nu$) if its unique minimal neighborhood decomposition is of size $\nu$. Moreover, this decomposition can be computed in linear time.
\end{definition}

We call the sets $T_1, \dots, T_\nu$ \textit{types}. Note that every type induces either a clique or an independent set in $G$ and two types are either joined by a complete bipartite graph or no edge between vertices of the two types is present in $G.$
We call a type which is a clique a \emph{clique type} and a type which is an independent set an \emph{independent type}.
Thus, we introduce the notion of a \emph{type graph} $T_{\TT}(G)$.
The vertices of $T_{\TT}(G)$ are the types $T_1, \dots, T_\nu$ and two types $T_i, T_j$ are joined by an edge if $T_i$ and $T_j$ are joined by a complete bipartite graph in $G$. If the decomposition $\TT$ is clear from the context, we omit the subscript $\TT$.

Observe that for every graph $G = (V,E)$ we have $\nd(G) \le 2^{\Gvc(G)} + \Gvc$.
Indeed, we can construct a decomposition $\TT$ witnessing this as follows.
Let $U$ be a vertex cover of $G$ with $|U| = \Gvc(G)$.
We put singleton $\{u\}$ to $\TT$ for every $u \in U$ and then we add sets $\{ v \in V \setminus U \mid N_G(v) = X \}$ to $\TT$ for every $X \subseteq U$.
It is easy to verify that $\TT$ is a neighborhood decomposition of $G$.
Thus, neighborhood diversity is a more general structural graph parameter than vertex cover, since the class of cliques has neighborhood diversity 1 while unbounded size vertex cover.

\subsection{Parameterized Complexity}
Let $\Sigma$ be a finite alphabet.
A \emph{parameterized language} is a language $P \subseteq \Sigma^* \times \mathbb{N}$.
The associated parameterized problem is then to decide whether the input $(x, k)$ belongs to $P$ or not; the value $k$ is the \emph{parameter}.
A parameterized language $P$ belongs to the class \FPT (is fixed-parameter tractable) if there is an algorithm deciding $P$ in $f(k) \cdot \operatorname{poly}(|x|)$ time, where $f \colon \mathbb{N} \to \mathbb{N}$ is a computable function.
A parameterized language $P$ belongs to the class \XP if there is an algorithm deciding $P$ in $|x|^{f(k)}$ time for a computable function $f \colon \mathbb{N} \to \mathbb{N}$.
Clearly, \FPT is a subclass of \XP.

Let $P$ and $Q$ be two parameterized languages.
A \emph{parameterized reduction} from $P$ to $Q$ is an algorithm that on input $(x,k)$ computes an instance $(y, \ell)$ in $f(k) \cdot \operatorname{poly}(|x|)$ time such that
\begin{itemize}
  \item $(x,k) \in P$ if and only if $(y, \ell) \in Q$ and
  \item $\ell \le g(k)$ for some computable function $g \colon \mathbb{N} \to \mathbb{N}$.
\end{itemize}
Roughly speaking, a parameterized problem $Q$ is \Wh{1} if there exists a parameterized from \textsc{Clique} parameterized by the solution size to~$Q$.
Under a common assumption $\FPT \neq \W{1}$ if a parameterized problem is \Wh{1}, then it is unlikely to be in \FPT.
For further detail please refer to e.g.~\cite{CyganFKLMPPS:15-FPTbook}.

\section{Graphs of Bounded Neighborhood Diversity}\label{sec:nd}
For graphs of bounded neighborhood diversity we prove two negative results (Theorems~\ref{thm:MSOGisWHwrtVC} and~\ref{thm:MSOLisWHwrtND}) and two positive results (Theorems~\ref{thm:JECKYLisFPTwrtND} and~\ref{thm:HYDEisXPwrtND}).

\subsection{Theorems~\ref{thm:MSOGisWHwrtVC} and \ref{thm:MSOLisWHwrtND}: \textsf{W[1]-hardness} of \MSOL and \MSOG}\label{sec:nd_hardness}

We begin with a definition of an auxiliary problem:

\probl{\textsc{Local Cardinality Constrained Subset} (\LCCSubset)}
{Graph ${G = (V,E)}$ with $|V|=n$ and a function ${f\colon V\to 2^{[0,n-1]}}$.}
{Find a set ${U\subseteq V}$ such that ${|U(v)| \in f(v)}$ for each vertex ${v\in V}$.}

\noindent Obviously \LCCSubset is equivalent to \MSOL with an empty formula $\varphi$.
We call an \LCCSubset instance \textit{uniform} if, on $G$ with neighborhood decomposition $\TT$, the demand function $f$ can be written as $f\colon \TT \rightarrow 2^{[0,n-1]}$, that is, the vertices of the same type have the same demand set.
We show that already uniform \LCCSubset is \Wh{1} by a reduction from the \Wh{1} \kMCC problem~\cite{CyganFKLMPPS:15-FPTbook}.

\probl{\kMCC\hfill {\em Parameter:} $k$}
{$k$-partite graph $G=(V_1 \dot\cup \cdots \dot\cup V_k,E)$, where $V_a$ is an independent set for every $a\in [k]$.}
{Find a clique of size $k$.}

\noindent We refer to a set $V_a$ as to a {\em colorclass} of $G$.
We may assume that all of the colorclasses are of the same size and that the number of edges between any two colorclasses is the same.
This follows from the reduction from $k$-\textsc{Clique} to \kMCC, where we take $k$ copies of the vertex set of the original graph and connect two vertices in the new instance if and only if they are in different copies and their pre-images are adjacent (i.e., the degree of every vertex to any other colorclass is the same as in the original graph).

Our proof is actually a simplified proof of \textsf{W[1]-hardness} for the {\sc Target Set Selection} problem~\cite{DvorakKT16} for the parameter neighborhood diversity.

\begin{theorem}\label{thm:LCCSubsetHardness}
\LCCSubset is \Wh{1} parameterized by the vertex cover number already in the case when $f(v) = \{0\}$ for all $v$ not belonging to the vertex cover.
\end{theorem}
\begin{proof}
Let $G = (V_1\cup \cdots\cup V_k, E)$ be the instance graph for \kMCC.
We naturally split the set of edges $E$ into sets $E_{\{a,b\}}$ by which we denote the edges between colorclasses $V_a$ and $V_b$.
We denote $n$ the (common) size of colorclasses in ~$G$, and we denote $m$ the number of edges between any two colorclasses.
Fix $N>n$, say $N = n^2$, and distinct $a,b\in [k]$.

\heading{Description of the reduction}
We numerate vertices in each color class $V_a$ for $a \in [k]$ using numbers in $[n]$, that is, we fix a bijection $\mu_a \colon V_a \to [n]$ for each $a \in [k]$.
We also numerate the edges between color classes $a$ and $b$ by numbers in $[m]$.
Let $\varepsilon_{\{ a,b \}} \colon E_{\{ a,b \}} \to [m]$ be the numeration function for distinct $a,b \in [k]$.
We set
\[
I_{ab} = \bigl\{ \mu_a(v) + N \cdot \varepsilon_{\{a,b\}}(e) \mid v\in e, e\in E_{\{a,b\}} \bigl\}\,.
\]
We build the graph $H$ using the following groups of vertices (refer to Figure~\ref{fig:LCCSubsetHardness}):
\begin{itemize}
  \item an independent set $S_a$ of size $n$ for each color class $V_a$ and set $f(v) = \{0\}$ for every $v\in S_a$,
  \item an independent set ${T_{\{a,b\}}}$ of size $mN$ for each edge set ${E_{\{a,b\}}}$, with $f(v) = {\{0\}}$ for every $v\in {T_{\{a,b\}}}$,
  \item a single vertex {Mult$_{\{a,b\}}$} with $f(\textrm{Mult}_{\{a,b\}}) = {\{tN\mid t\in [m]\}}$ for each $\{ a, b \} \in \binom{[k]}{2}$, and
  \item a single vertex {Inc$_{ab}$} with $f(\textrm{Inc}_{ab}) = {I_{ab}}$ for each $a,b \in [k]$ with $a \neq b$.
\end{itemize}
Finally, we add all possible edges between ${S_a}$ and {Inc$_{ab}$}, between {Inc$_{ab}$} and ${T_{\{a,b\}}}$, and between ${T_{\{a,b\}}}$ and {Mult$_{\{a,b\}}$}, thus forming complete bipartite subgraphs between the respective sets of vertices.
It is straightforward to check that the $\binom{k}{2}$ vertices {Mult$_{\{a,b\}}$} together with $k(k-1)$ vertices {Inc$_{_{ab}}$} form a vertex cover of $H$.
It follows that $\Gvc{(H)} = \binom{k}{2} + k(k-1)$.
For an overview of the reduction please refer to Figure~\ref{fig:LCCSubsetHardness}.

\noindent\textbf{Correctness of the reduction.}
Suppose there is a clique of size $k$ in $G$ with vertex set $\{v_1,\ldots,v_k\}$.
We assume that $v_i \in V_i$ for all $i \in [k]$.
We select $\mu_a(v_a)$ vertices in the set $S_a$ and $N \cdot \varepsilon_{\{a,b\}}(\{v_a,v_b\})$ vertices in the set $T_{\{a,b\}}$ for all distinct $a,b \in [k]$.
It is straightforward to check that this is a solution respecting the demands in~$H$.

For the opposite direction suppose there is a solution $U$ respecting demands in~$H$.
First note that none of vertices {Mult$_{\{a,b\}}$}, {Inc$_{ab}$} is selected as their neighborhood demands are set to~$0$.
Denote $s_a = |U\cap S_a|$ and $t_{\{a,b\}} = |T_{\{a,b\}} \cap U|$.
Now observe that since the demand of vertex {Mult$_{\{a,b\}}$} is fulfilled, there are $t_{\{a,b\}} = tN$ vertices for some $t \in [m]$.
Let $e_{ab}$ denote the edge in $E_{\{a,b\}}$ with numeration $\varepsilon_{\{a,b\}}(e_{ab}) = t$ for each distinct $a,b \in [k]$.
Let $v_a$ denote the vertex in $V_a$ with numeration $\mu_a(v_a) = s_a$ for every $a \in [k]$.
We now want to prove that respecting demands of {Inc$_{ab}$} vertices implies that $\{v_a,v_b\} \in E_{\{a,b\}}$ for all $a,b \in [k]$, that is, the vertices $\{ v_a \mid a \in [k]\}$ for a (multicolored) clique in~$G$.
Since the demand of vertex {Inc$_{ab}$} is fulfilled, it follows that the vertex $v_a$ must be incident to the edge $e_{ab}$.
Symmetrically, since the demand for the vertex {Inc$_{ba}$} is fulfilled, we get that the vertex $v_b$ is incident to the edge $e_{ab}$.
Combining all of these together we infer that $\{ v_1, \ldots, v_k \}$ defined in this way form a clique in the graph $G$.
This concludes the proof.
\end{proof}

\begin{figure}[bt]
  \begin{center}
\begin{tikzpicture}[node distance=2.5cm]
  \tikzstyle{bag} = [circle, draw, inner sep=1pt]
  \tikzstyle{edge} = [black, ultra thick]
  \tikzset{font={\fontsize{11pt}{12}\selectfont}}

\node[bag, label={90:$S_a$}, label={270:$0$}] (Va) {$n$};
\node[bag, right of=Va, label={90:Inc$_{ab}$}, label={270:$I_{ab}$}] (Iab) {$1$};
\node[bag, right of=Iab, label={90:$T_{\{a,b\}}$}, label={270:$0$}] (Eab) {$mN$};
\node[bag, right of=Eab, label={90:Mult$_{\{a,b\}}$},label={270:$\{tN\mid t\in [m]\}$}] (nas) {$1$};

\draw[edge] (Va) -- (Iab) -- (Eab) -- (nas);
\end{tikzpicture}
  \end{center}
  \caption{An~overview of the decomposition of a~gadget used in the~proof of~Theorem~\ref{thm:LCCSubsetHardness}. Numbers inside nodes denote the number of vertices in the independent set represented by the node. Below each node a~description of the~respective set of admissible numbers is shown.}
  \label{fig:LCCSubsetHardness}
\end{figure}

Note that Theorem~\ref{thm:MSOLisWHwrtND} follows easily from Theorem~\ref{thm:LCCSubsetHardness}, since the \LCCSubset problem is expressed by an empty \MSOLLin formula.
Furthermore, we get the following consequence of the presented proof.

\begin{corollary}\label{cor:UniformIndependentLCCSubsetHardnessForND}
  \LCCSubset is \Wh{1} parameterized by neighborhood diversity even if
  \begin{itemize}
    \item the instance of \LCCSubset is uniform for the given decomposition and
    \item all of the types in the given decomposition are independent sets.
  \end{itemize}
\end{corollary}
\begin{proof}
  Observe that neighborhood diversity of the graph resulting from the construction presented in the proof of Theorem~\ref{thm:LCCSubsetHardness} has neighborhood diversity at most \( \binom{k}{2} + k(k-1) + \binom{k}{2} + k\).
  To see this note that we can introduce a type with one vertex for every vertex in the vertex cover of the constructed graph (recall there are \( \binom{k}{2} + k(k-1) \) many of these).
  Furthermore, all of the vertices in (independent sets) {Inc$_{ab}$} (for distinct fixed $a,b \in [k]$) have the same neighborhood (in the vertex cover) and thus form a single type; which is an independent set and recall that for all such vertices~$v$ we have $f(v) = 0$.
  The same holds for vertices in {Mult$_{\{a,b\}}$} (again for fixed $a,b \in [k]$).
  This finishes the proof.
\end{proof}

As we mentioned in the introduction, our hardness result has consequences for the hardness of the \textsc{Multidemmand Set Multicover} problem.
We now use Corollary~\ref{cor:UniformIndependentLCCSubsetHardnessForND} as the basis of our reduction.

\begin{proof}[Proof of Theorem~\ref{thm:msc_wh}]
  \sloppy
Given a uniform instance of \LCCSubset on a graph $G$ with \mbox{$\nd(G) = \nu$} with every type being an independent set, let $U = [\nu]$, $\mathcal{F} = \{N(v) \mid \forall v \in T(G)\}$ and let $D_i = f(i)$.
Now, if there exists an $r \in [n]$ such that $(U, \mathcal{F}, (D_1, \dots, D_\nu), r)$ is a \textsc{Yes} instance of \textsc{Multidemand Set Multicover}, then the given \LCCSubset instance is a \textsc{Yes} instance, and otherwise it is a \textsc{No} instance.
Consequently, one can test all of the possible values of~$r$ in polynomial time and thus obtain the answer for the graph~$G$.
\end{proof}
\fussy %
Having showed the hardness of \MSOL parameterized by $\nd(G)$, let us now turn our attention to the proof of Theorem~\ref{thm:MSOGisWHwrtVC}, i.e., hardness of \MSOG parameterized by $\nd(G)$.

\begin{proof}[Proof of Theorem~\ref{thm:MSOGisWHwrtVC}]
Let $(G = (V, E), f, k)$ be an instance of the \LCCSubset problem parameterized by the vertex cover number resulting from Theorem~\ref{thm:LCCSubsetHardness}.
Let $C\subseteq V$ be the vertex cover in $G$.
Note that it follows from the proof of Theorem~\ref{thm:LCCSubsetHardness} that we may assume that the independent set $V\setminus C$ is divided into $\Oh(k)$ groups, where each group shares the neighborhood in $C$.
Observe further that indeed the graph $G$ is bipartite (i.e., the set $C$ is also an independent set), in particular, the largest clique subgraph of $G$ is of size 2.

By Theorem~\ref{thm:LCCSubsetHardness} we know that it is \Wh{1} to find a subset $X \subseteq V\setminus C$ such that $|X(v)| \in f(v)$ for all $v \in C$.
Our goal now is to build an \MSOG formula expressing exactly this.

First we take $G$ and construct a graph $G'$ by, for each $v\in C$, attaching a $K_{2+\eta(v)}$ to $N(v)$, where $\eta\colon C\to [k]$ is a bijective mapping.
We will call the clique $K_{2+\eta(v)}$ a \textit{marker} because it will allow us to recognize exactly the vertices of $N(v)$.
Note that markers are the only cliques present in $G'$ of size at least 3.
Note further that by this we have added $\Oh(k)$ cliques of size $\Oh(k)$ and thus the resulting graph has vertex cover of size $\Oh(k^2)$.

Let us describe some auxiliary formulae which we then use to define the desired formula~$\varphi$.
We reserve $X$ for the set that will represent the set $X$ from the \LCCSubset problem.

\begin{itemize}

  \item {\bf $Z$ is a clique:}
  \begin{flalign*}
  &\qquad \qquad \texttt{clique}(Z) := (\forall x,y \in Z)(x \neq y \implies xy \in E)&
  \end{flalign*}

  \item {\bf $u$ and $v$ are of the same neighborhood type:}
  \begin{flalign*}
  &\qquad \qquad \texttt{same}(u,v) := (\forall w\in V)(w=u \lor w=v \lor (wu\in E \iff wv\in E))&
  \end{flalign*}

  \item {\bf $Z$ is a type:}
  \begin{flalign*}
  &\qquad \qquad\texttt{type}(Z) := (Z\neq\emptyset) \land (\forall u,v\in Z)(\texttt{same}(u,v)) \land (\forall u\in Z,v\notin Z)(\lnot\texttt{same}(u,v)) &
  \end{flalign*}

  \item {\bf $Z$ is $\eta(v)$-th marker:}
  \begin{flalign*}
  &\qquad \qquad \texttt{marker}_v(Z) := (|Z| = 2+\eta(v)) \land \texttt{clique}(Z) \land \texttt{type}(Z)&
  \end{flalign*}

  \item {\bf $Z$ is $N(v)$:}
  \begin{flalign*}
  &\qquad \qquad
  \texttt{neigh}_v(Z) := \texttt{type}(Z) \land (\exists Q \subseteq V)(\texttt{marker}_v(Q) \land (\forall u\in Z, w\in Q)(uw\in E))&
  \end{flalign*}

  \item {\bf $Z$ is exactly $X_v$:}
  \begin{flalign*}
  &\qquad \qquad \texttt{sel-neigh}_v(Z,X) := (\exists Z_v)(\texttt{neigh}_v(Z_v) \land Z = Z_v\cap X )&
  \end{flalign*}
\end{itemize}
Now $\varphi(X, (X_v)_{v\in C}) := \bigwedge_{v\in C} \bigl(\texttt{sel-neigh}_v(X_v,X) \land |X_v| \in f(v) \bigr).$
\end{proof}

\subsection{Theorem~\ref{thm:JECKYLisFPTwrtND}: \FPT algorithm for \MSOGLLin on neighborhood diversity}\label{sec:nd_fptalgo}
Essentially, we are modifying the algorithm of Ganian and Obdržálek~\cite{GO:13} for \MSOGLin model checking so that it can deal with the additional constraints introduced by \MSOLLin.
We use integer linear programming (ILP).
By the result of Lenstra~\cite{Lenstra83} ILP can be solved in \FPT-time parameterized by the number of integral variables.

\subsubsection{Signatures and Shapes}
Before we move on to proving Theorem~\ref{thm:JECKYLisFPTwrtND} we first need to introduce some notation.
\begin{definition}
Let $\varphi$ be an \MSOGLLin formula with free set variables $X_1, \ldots, X_\ell$, let $G = (V,E)$ be a graph with $\nd(G)=\nu$ and types $T_1, \ldots, T_\nu$, and let $\mu \colon \{X_1, \ldots, X_\ell\} \to 2^V$ be a variable assignment.
The \emph{signature of $\mu$} is the mapping $S_\mu \colon [\nu] \times 2^{[\ell]} \to \mathbb{N}$ defined by
\[
S_\mu(j, I) = \left| \bigcap_{i\in I} \mu(X_i) \cap T_j \right|
\]
for $j \in [\nu]$ and $I \subseteq [\ell]$.
\end{definition}
Clearly, if we have two variable assignments $\mu$ and $\mu'$ with the same signature, then $G,\mu \models \varphi$ if and only if $G,\mu' \models \varphi$.

However, for \MSO formulae and graphs of bounded neighborhood diversity, much more is true.
Informally speaking, the formula cannot distinguish between two cardinalities if both of them are large.
This is formally stated in the next lemma, which is a direct consequence of~\cite[Lemma 5]{Lampis:12}:
\begin{lemma}\label{lem:assignment_equivalence}
Let $\varphi$ be an $\MSO$ formula with free set variables $X_1, \ldots, X_\ell$ that has $q_S$ set quantifiers and $q_e$ element quantifiers.
Let $G$ be a graph with $\nd(G)=\nu$ and let $t = 2^{q_S}\cdot q_e$.
Suppose that $\mu$ and $\mu'$ are two variable assignments such that for every $I \subseteq [\ell]$, $j \in [\nu]$ we have either
	\begin{itemize}
		\item $S_\mu(j,I) = S_{\mu'}(j,I)$, or
		\item both $S_\mu(j,I), S_{\mu'}(j,I) > t$.
	\end{itemize}
Then $G, \mu \models \varphi$ if and only if $G, \mu' \models \varphi$.
\end{lemma}

The last lemma leads to the following definition.

\begin{definition}[Shape]\label{def:shape}
Let $\varphi$, $G$, and $t$ be as before.
A \emph{shape} of a variable assignment $\mu \colon \{X_1, \ldots, X_\ell\} \to 2^V$ is the mapping $sh_\mu \colon [\nu] \times 2^{[\ell]} \to [0,t] \cup \{\many\}$ defined by
\[
	sh_\mu(j,I) =
	\begin{cases}
		S_\mu(j,I) & \text{ if $S_\mu(j,I) \leq t$} \\
		\many & \text{ if $S_\mu(j,I) > t$} \\
	\end{cases} \,.
\]
\end{definition}
Since $t$ depends only on the formula $\varphi$, the total number of shapes can be bounded by some function of $||\varphi||$ and $\nd(G)$.
Note that there are mappings from $[\nu] \times 2^{[\ell]}$ to $[0,t] \cup \{\many\}$ that do not correspond to a shape of any variable assignment $\mu$ for a particular graph $G$. For example, if $sh(j,I) = \many$ for some $j$ and $I$ but $|T_j| < t$, clearly there is no assignment of such a shape $sh$.

It is worth noting that Lemma~\ref{lem:assignment_equivalence} cannot be used directly, as the global linear constraints allow us to distinguish small differences in cardinalities, even if the cardinalities are large; consider for example the constraint $\left[|X_1| = |X_2| + 1\right]$.
We use the approach outlined in Subsection~\ref{subsec:mso_and_its_extensions}, Pre-evaluations.
This approach relies on Definitions~\ref{def:pre-evaluation} and~\ref{def:complying_pre-evaluation}.
We simply guess all possible outcomes of the (global) cardinality constraints (the number of such outcomes is clearly bounded by $2^{||\varphi||}$) and later ensure that our assignment obeys those constraints by an Integer Linear Program.

\begin{definition}
A shape $sh$ is \emph{admissible} with respect to a pre-evaluation $\beta$ if for any variable assignment $\mu$ of the shape $sh$ we have $G, \mu \models \beta(\varphi)$.
\end{definition}

\subsubsection{Unifying Local Linear Constraints}
Here we show how to change the local linear constraints and the neighborhood diversity decomposition in a such way that
\begin{itemize}
  \item the new instance is equivalent to the former one,
  \item the size of the new decomposition is bounded in terms of $\nd(G)$ and $\| \varphi \|$, and
  \item vertices of the same type in the newly obtained neighborhood diversity decompostion also have exactly the same local linear constraints.
\end{itemize}
This, in turn, allows us to prove the main theorem in a much simpler setting.

\paragraph{Single Local Cardinality Constraint}
We first show how to alter the given decompostion with respect to one local cardinality constraint $\alpha$ which is the core of the represented reduction.
Then, we show how this can be used to alter the neighborhood diversity decomposition when more local cardinality constraints $( \alpha_1, \ldots, \alpha_\ell )$ are given.

Let $G$ be a graph and $\TT$ its neighborhood diversity decomposition.
A type $T\in \TT$ is said to be \emph{nonuniform} with respect to local linear cardinality constraint $\alpha$ if there exist vertices $u,v\in T$ with $\alpha(u)\neq\alpha(v)$, otherwise $T$ is said to be \emph{uniform}.
As already mentioned, the purpose of this section is to alter the given instance into an equivalent uniform one.
In order to be able to do so we have to change the neighborhood diversity decomposition (i.e., the type graph).
A neighborhood diversity decomposition $\hat\TT$ is a \emph{refinement} of $\TT$ if for every $\hat{T} \in \hat\TT$ there exists a type $T \in \TT$ such that $\hat{T} \subseteq T$.
We define $\nu_{\alpha}(\TT)$ as the number of nonuniform types in $\TT$ with respect to $\alpha$.

\begin{proposition}\label{prop:solutionInType}
Let $G = (V,E)$ be a graph and let $\TT$ be a neighborhood diversity decomposition of~$G$.
For every $T\in\TT$ and for every $X \subseteq V$ there exists a nonnegative integer $z$ such that for every vertex $w \in T$
\begin{itemize}
  \item it holds that $|X(w)| = z$ if $T$ is an independent set and
  \item it holds that $|X(w)| \in \{z,z+1\}$ if $T$ is a clique.
\end{itemize}
\end{proposition}
\begin{proof}
First assume that $T$ is an independent set, then $N(v) = N(w)$ for all $v,w\in T$.

For the second case assume that $T$ is a clique and let $M = N(v)\setminus\{w\}$.
Now $N(v) = M\cup \{w\}$ and $N(w) = M\cup\{v\}$ and, as the number $|M\cap X|$ contributes to both $X(v)$ and $X(w)$, $\big||X(v)| - |X(w)|\big| \leq 1$ must hold.
\end{proof}

Now, we show how to refine the neighborhood diversity decomposition with respect to one local linear cardinality constraint.
\begin{lemma}\label{lem:local_refinement}
Given a graph $G=(V,E)$, neighborhood diversity decomposition $\TT$ of $G$, and local linear cardinality constraint $\alpha$.
Let $T\in \TT$ be a nonuniform type.
There exists a partition $\TT'$ of $T$ and local linear cardinality constraint $\alpha'$ such that the following holds
\begin{enumerate}
  \item $|\TT'| \le 4$,
  \item if $T' \in \TT\setminus \{T\}$ is a uniform type with respect to~$\alpha$, then $T'$ remains uniform with respect to $\alpha'$,
  \item $\nu_{\alpha'}((\TT\setminus \{T\})\cup \TT') < \nu_{\alpha}(\TT)$, and
  \item for each $X\subseteq V$, $X$ satisfies $\alpha$ if and only if $X$ satisfies $\alpha'$.
\end{enumerate}
\end{lemma}
\begin{proof}
Let us first argue about an independent type $T$.
In this case it suffices to set $\alpha'(u) = \bigcap_{v\in T}\alpha(v)$ for each $u\in T$.
Now $X\subseteq V$ satisfies $\alpha$ if and only if $X$ satisfies $\alpha'$ as the value $|N(T)\cap X|$ has to be the same for all vertices of $T$ and thus has to be in $\alpha'(v)$ for $v\in T$ by Proposition~\ref{prop:solutionInType}.

Let $T$ be a clique type of $\TT$.
We define
\begin{align}
\texttt{l} &= \max_{v\in T} \min \alpha(v) \text{ and} \label{eq:ldef} \\
\texttt{u} &= \min_{v\in T} \max \alpha(v) \enspace . \label{eq:udef}
\end{align}
If $\texttt{u} \le \texttt{l} - 2$, then $\alpha$ cannot be satisfied by Proposition~\ref{prop:solutionInType}.
Define the new local linear constraint $\alpha'(v) = \alpha(v)\cap [\texttt{l}-1 , \texttt{u} +1]$ for every $v\in T$ and define $\alpha'(v) = \alpha(v)$ for every $v\in V\setminus T$.
We get that:
\begin{itemize}
	\item $\alpha'(v) \subseteq [\texttt{l} - 1, \texttt{u} + 1]$ for each $v\in T$ and
	\item $[\texttt{l}, \texttt{u}] \subseteq \alpha'(v)$ for each $v\in T$.
\end{itemize}
This yields at most four possibilities for $\alpha'(v)$ for $v \in T$; namely $\alpha'(v)$ is one of the sets $[\texttt{l}-1, \texttt{u}], [\texttt{l}-1, \texttt{u}+1], [\texttt{l}, \texttt{u}]$, or $[\texttt{l}, \texttt{u}+1]$.
We can refine $T$ into at most $4$ subtypes such that all the vertices of a subtype of $T$ have the same $\alpha'(v)$.
As all newly introduced types are uniform with respect to~$\alpha'$, we have replaced a nonuniform type $T$ with at most $4$ uniform types (while we have kept all other types untouched).
We have proven (1)--(3); in order to prove (4) we use the following claim.

Clearly, all subtypes of~$T$ are uniform with respect to $\alpha'$ and for each $X\subseteq V$ it holds that if $X$ satisfies $\alpha'$, then $X$ satisfies $\alpha$, since $\alpha'(v) \subseteq \alpha(v)$ for every vertex $v \in V$.
Thus, it remains to show the converse.

\begin{claim} \label{cl:bla}
Let $p\in [n]$ and let $\texttt{l}$ be defined as in~\eqref{eq:ldef}.
If there exists $v\in T$ such that the following conditions are fulfilled
\begin{itemize}
  \item
   $p\in\alpha(v)$ and
  \item
   $p\le \texttt{l}-2$,
\end{itemize}
then for each $X$ satisfying $\alpha$ it holds that $p\neq |X(v)|$.
\end{claim}
\begin{proof}[{Proof of Claim~\ref{cl:bla}}]
Let $z$ be as in Proposition~\ref{prop:solutionInType}, that is, each $w\in T$ must have $z$ or $z+1$ in $\alpha(w)$.
Suppose for a contradiction that $|X(v)| = p$ and let $s$ be a vertex with $\alpha(s) \subseteq \{\texttt{l},\ldots\}$ (such $s$ exists from the definition of $\texttt{l}$).
As $p \le \texttt{l}-2$, it follows that $X$ cannot satisfy $\alpha(s)$.
There are two possible options $\{p-1,p\}$ and $\{p,p+1\}$ for the value of $z$ from Proposition~\ref{prop:solutionInType}.
Observe that $\{p-1,p,p+1\}\cap \alpha(s) = \emptyset$.
This finishes the proof of the claim.
\end{proof}

Let $X\subseteq V$ satisfy $\alpha$; otherwise there is nothing to prove.
By the above claim and its symmetric version for $p\ge \texttt{u}+2$ it follows that $\texttt{l} - 1\le X(v)\le \texttt{u}+1$.
By the definition of $\alpha'$ it follows that $X$ satisfies $\alpha'$.
\end{proof}

\paragraph{Local Cardinality Constraints}
We now apply the above lemma to all local cardinality constraints $(\alpha_1, \ldots, \alpha_\ell)$ in the given instance.
\begin{lemma}\label{lem:refinement}
Given a graph $G=(V,E)$ with $\nd(G)=\nu$ and with local linear cardinality constraints $(\alpha_1, \ldots, \alpha_\ell)$, there exists a neighborhood decomposition $\TT$ of $G$ of size at most $\nu 4^\ell$ and local linear cardinality constraints $(\alpha'_1, \ldots, \alpha'_\ell)$ such that:
\begin{itemize}
\item each type $T \in \TT$ is uniform with respect to $\alpha_i$ for all $i \in [\ell]$, and,
\item for each $(X_1, \dots, X_\ell) \subseteq V^\ell$, $X_i$ satisfies $\alpha_i$ for all $i\in[\ell]$ if and only if $X_i$ satisfies $\alpha'_i$ for all $i\in[\ell]$.
\end{itemize}
\end{lemma}
\begin{proof}
The proof goes by repeatedly applying Lemma~\ref{lem:local_refinement}. We start with the neighborhood decomposition $\hat{\TT}$ of size $\nu$ that is guaranteed by $\nd(G)=\nu$, and with the local linear cardinality constraints $(\alpha_1, \ldots, \alpha_\ell)$.

First let $i=1$, and as long as there is a type $T \in \hat{\TT}$ that is nonuniform with respect to $\alpha_1$ we do the following.
We apply Lemma~\ref{lem:local_refinement} to the type $T$, the local linear cardinality constraint $\alpha'_1$ and decomposition $\TT'$ resulting from the previous application of the lemma (using $\alpha_1$ and $\hat{\TT}$ in the first iteration).
Note that in such an iteration we leave all of the other local cardinality constraints $(\alpha_2, \ldots, \alpha_\ell)$ intact.
Clearly, after we are done we have a neighborhood decomposition $\TT'$ of size at most $4 \cdot \nu$ and local linear cardinality constraints $(\alpha'_1, \alpha_2, \ldots, \alpha_\ell$ such that every type $T \in \TT'$ is uniform with respect to $\alpha'_1$.

Then, continuing with $i\in[2, \ell]$, we do the same, finally resulting in a decomposition $\TT$ of size $\nu4^\ell$ and local linear cardinality constraints $(\alpha'_1, \dots, \alpha'_\ell)$.
Note that the only side effect of an invocation of Lemma~\ref{lem:local_refinement} is a refinement of some type $T$ and observe that if~$T$ is uniform with respect to $\alpha'_j$, then so is any of its subtypes in the refinement.
Consequently, every type $T \in \TT$ is uniform with respect to $\alpha'_i$ for all $i \in [\ell]$.
\end{proof}

\subsubsection*{Uniform Instance}
Let $G$ be a graph and let $(\alpha_1, \ldots, \alpha_\ell)$ be local cardinality constraints.
For a given neighborhood diversity decomposition $\TT$ we say that $(\alpha_1, \ldots, \alpha_\ell)$ are \emph{uniform} if $\nu_{\alpha_i}(T) = 0$ for all $T\in\TT$ and all $i \in [\ell]$.

\subsubsection{Uniform Instance Theorem}
\begin{theorem}\label{thm:ndUnifGLLin}
There exists an algorithm that, for given an \MSOGLLin formula $\varphi$ with free set variables $X_1, \ldots, X_\ell$, graph $G = (V,E)$ with neighborhood diversity decomposition $\TT$, and uniform local linear constraints $(\alpha_1, \ldots, \alpha_\ell)$ decides whether there exists an assignment $\mu$ such that $G, \mu \models \varphi$ and $|\mu(X_i) \cap N(v)| \in \alpha_i(v)$ for every $v \in V$ and every $i \in [\ell]$.
The algorithm terminates in time $f(\|\varphi\|,\nu) n^{\Oh(1)}$ for some computable function $f$.
Moreover, if such an assignment exists, the algorithm outputs one.
\end{theorem}
\begin{proof}
Let $\nu$ be the size of $\TT$.

The algorithm works as follows.
For every pre-evaluation function $\beta$ and every mapping $sh \colon [\nu] \times 2^{[\ell]} \to [0,t] \cup \{\many\}$, we test whether $sh$ is admissible.
This can be done by picking an arbitrary variable assignment $\mu$ of shape $sh$ (if there exists such an assignment) and testing whether $G, \mu \models \beta(\varphi)$ by an \FPT model checking algorithm for \MSO formulae~\cite{Lampis:12}.

If the shape $sh$ is admissible with respect to $\beta$, we need to find a variable assignment $\mu$ such that
\begin{itemize}
	\item $\mu$ complies with $\beta$,
	\item $\mu$ has shape $sh$, and
	\item $\mu$ satisfies the local linear constraints.
\end{itemize}
We find such an assignment $\mu$ by the following integer linear program (which is infeasible if $\mu$ does not exist).
We begin the description of the integer linear program with a description of all its variables:
\begin{itemize}
	\item for every $I \subseteq [\ell], j \in [\nu]$, we introduce an integer variable $x_I^j$ (these correspond to $S_\mu(j,I)$ of the variable assignment $\mu$ we are about to find),
	\item for every $i \in [\ell]$, $j \in [\nu]$, we introduce an auxiliary variable $y_i^j$ corresponding to $|\mu(X_i) \cap T_j|$, and
	\item for every $i \in [\ell]$, we add an auxiliary variable $z_i$ corresponding to $|\mu(X_i)|$ (technically variables $y_i^j$ and $z_i$ are redundant as they are projections of the $x$ variables, but they will simplify the presentation).
\end{itemize}
To ensure that $\mu$ has the required properties, we add the following constraints:
\begin{align*}
	\label{eq:consistency}\sum_{I \subseteq [\ell]} x_I^j &= |T_j|  & \mbox{for every $j \in [\nu]$} & \tag{0}\\
	\label{eq:aux_y}y_i^j &= \sum_{\{i\} \subseteq I \subseteq [\ell]} x_I^j & \mbox{for every $j \in [\nu]$ and every $i \in [\ell]$} &\tag{a1}\\
	\label{eq:aux_z}z_i &= \sum_{j=1}^\nu y_i^j & \mbox{for every $i \in [\ell]$} & \tag{a2} \\
	\label{eq:shape_small}x_I^j &= sh(I,j) & \mbox{for every $j \in [\nu], I \subseteq [\ell]$ such that $sh(I,j) \neq \many$} & \tag{sh1}\\
	\label{eq:shape_large}x_I^j &> t  & \mbox{for every $j \in [\nu], I \subseteq [\ell]$ such that $sh(I,j) = \many$} & \tag{sh2}\\
\end{align*}
We note that if the variables $x$ are integral, then all auxiliary variables are integral as well (as we obtain them only by summing up integers).
Thus, these variables can be real and thus do not contribute to the total number of integral variables (i.e., only the $x$ variables do).
The constraints \eqref{eq:consistency} ensure that variables $x_I^j$ encode a variable assignment~$\mu(x)$ for the graph $G$.
This is as follows:
There are exactly $x_I^j$ vertices form type $T_j$ in the set $\bigcap_{i \in I} \mu(X_i)$.
Now, it is not hard to see that constraints \eqref{eq:consistency} ensure that every vertex of type $T_j$ is placed (possibly to none of $X_i$'s).
The constraints \eqref{eq:aux_y} and \eqref{eq:aux_z} set auxiliary variables $y_i^j$ and $z_i$ to the desired values.
The constraints \eqref{eq:shape_small} and \eqref{eq:shape_large} guarantee that $\mu(x)$ has the shape~$sh$.

We for convenience denote the $i$-th linear cardinality constraint for type $T_j$ by $\alpha_{i,j}$ and furthermore we set $\lb^i_j := \min \{ \alpha_{i,j}(v) \mid v \in T_j \}$ and $\ub^i_j := \max \{ \alpha_{i,j}(v) \mid v \in T_j \}$.

If $T_j$ is an independent type, we need to ensure that for every $v \in T_j$ we have $|\mu(X_i) \cap N(v)| \in \alpha_{i,j}(v)$.
It is easy to see that the quantity $|\mu(X_i) \cap N(v)|$ is the same for every $v \in T_j$ and it can be expressed as
\[
  \sum_{j': \{j',j\} \in E(T_G)} \left|\mu(X_i) \cap T_{j'}\right|\,.
\]
By the definition of auxiliary variables $y_I^j$, we have that $|\mu(X_i) \cap T_{j'}| = y_i^{j'}$, so the local linear condition for the variable $X_i$ can be rewritten as
\begin{equation}
\label{eq:loc_indep} \lb^i_j \leq \sum_{j': \{j',j\} \in E(T_G)} y_i^j \leq \ub^i_j\,. \tag{lli}
\end{equation}

If $T_j$ is a clique type, we have to be slightly more careful, since the quantity $|\mu(X_i) \cap N(v)|$ depends on whether $v$ is in $\mu(X_i)$ or not.
The set $N(v)$ does not include $v$ itself, so if $|\mu(X_i) \cap N(v')| = |\mu(X_i) \cap N(v)| + 1$ for every $v \in T_j \cap \mu(X_i)$, then $v' \in T_j \setminus \mu(X_i)$.
Similarly as before, we have equations
\[
|\mu(X_i) \cap N(v)| = \left( \sum_{j': \{j',j\} \in E(T_G)} \left|\mu(X_i) \cap T_{j'}\right|\right) - 1
\]
for $v \in T_j \cap \mu(X_i)$, and
\[
\left|\mu(X_i) \cap N(v)\right| = \sum_{j': \{j',j\} \in E(T_G)} \left|\mu(X_i) \cap T_{j'}\right|
\]
for $v \in T_j \setminus \mu(X_i)$.

This means that we need to add the constraint
\begin{equation}
\label{eq:loc_clique_1}
\lb^i_j \leq \sum_{j': \{j',j\} \in E(T_G)} y_i^{j'} \leq \ub^i_j \tag{llc1} \\
\end{equation}
if $|\mu(X_i) \cap T_j|\geq 1$ and add the constraint
\begin{equation}
\label{eq:loc_clique_2}
\lb^i_j \leq \sum_{j': \{j',j\} \in E(T_G)} y_i^{j'} - 1 \leq \ub^i_j \tag{llc2} \\
\end{equation}
if $|T_j \setminus \mu(X_i)| \geq 1$.

Fortunately, we can deduce whether the conditions $|\mu(X_i) \cap T_j| \geq 1$ or $|T_j \setminus \mu(X_i)| \geq 1$ hold already from the shape $sh$.
If we have
\begin{equation*}
	\label{eq:shape_condition_1}
	\sum_{I\,:\, I \ni i} sh(I,j) > 0 \qquad\qquad\qquad\qquad \left( \forall i \in [\ell] \right) \,,  \tag{cllc1}
\end{equation*}
then $\mu(X_i)$ necessarily intersect $T_j$,
whereas if we have
\begin{equation*}
	\label{eq:shape_condition_2}
	\sum_{I\,:\, I \not\ni i} sh(I,j) > 0 \qquad\qquad\qquad\qquad \left( \forall i \in [\ell] \right) \,,  \tag{cllc2}
\end{equation*}
then there exists vertex in $T_j \setminus \mu(X_i)$.

This means that the local linear constraints for type $T_j$ and variable $X_i$ can be enforced by adding constraint \eqref{eq:loc_clique_1} if \eqref{eq:shape_condition_1} holds, and by adding constraint \eqref{eq:loc_clique_2} if \eqref{eq:shape_condition_2} holds.

It remains to add all of the constraints arising from the pre-evaluation~$\beta$.
However, this is not a problem, since we know that every such condition is linear and as such can be easily added to the so far constructed MILP.
Note that any constraint pre-evaluated in~$\beta$ is of the form $\sum_{i = 1}^\ell c_i \cdot |X_i| \le b$.
We can assume that all constants $c_i$ are integers, since otherwise we can multiply all of them by greatest common divisor.
Furthermore, observe that $b$ can be assumed to be integral as well, since the left-hand side is now integral and thus upper-bounding it by $b$ is the same as upper-bounding it by $\left\lfloor b \right\rfloor$.
Now, based on~$\beta$ we either add the constraint $\sum_{i = 1}^\ell c_i \cdot |X_i| \le b$ or the constraint $\sum_{i = 1}^\ell c_i \cdot |X_i| \ge b + 1$ to the above constructed MILP.

Let us turn our attention to the analysis of the running time of the algorithm.
There are at most
\begin{itemize}
  \item $(t+2)^{\nu}$ different shapes and
  \item $2^{||\varphi||}$ pre-evaluation functions.
\end{itemize}
Since $t$ depends only on the number of quantifiers in the formula $\varphi$, both numbers can be bounded by a function of $||\varphi||$ and $\nu$.
For each such combination of a shape and a pre-evaluation, we construct an ILP with $\nu2^\ell$ integer variables, so this ILP can be solved in time \FPT time with respect to $||\varphi||$ and $\nu$ by the aforementioned result of Lenstra~\cite{Lenstra83}.
\end{proof}

\begin{proof}[Proof of Theorem~\ref{thm:JECKYLisFPTwrtND}]
The theorem is a simple consequence of Theorem~\ref{thm:ndUnifGLLin} and Lemma~\ref{lem:refinement}.
\end{proof}

\subsection{Theorem~\ref{thm:HYDEisXPwrtND}: \XP algorithm for \MSOGL}\label{sec:nd_xpalgo}
The main idea behind the proof of Theorem~\ref{thm:HYDEisXPwrtND} is as follows.
We use the advantage of \XP time to guess all sizes $b_I^T$ of the sets $T \cap \bigcap_{i \in I} \mu(X_i)$ for a possible assignment~$\mu$, for every type $T \in \mathcal{T}$, and for every $I \subseteq [\ell]$.
Clearly, the number of possible assignments of $b_I^T$ can be upper-bounded by $n^{|\TT| 2^\ell}$.
This immediately allows us to verify global cardinality constraints and in particular $\varphi$.
Note that at this point this is possible, since the vertices in $T$ are equivalent with respect to \MSO logic and thus the \MSO-model-checking algorithm of Lampis~\cite{Lampis:13} can be used to verify $\varphi$ on the (now labeled) graph $G$.
Here the labeling represents the assignment~$\mu$ respecting the guessed values~$b_I^T$.
If $\varphi$ is satisfied, we proceed in a similar way as in the proof of Theorem~\ref{thm:JECKYLisFPTwrtND}.
Now, we have to check whether there exists an assignment $\mu$ that obeys all local cardinality constraints for each vertex in $G$.
Observe that this can be done locally, as the only thing that matters in the neighboring types is the number of vertices in the sets $X_1,\ldots, X_\ell$ (especially it is independent of the actually selected vertices).
This however, results in a computation of lower- and upper-bounds on $b_{I}^T$ in a similar but simpler way to Lemma~\ref{lem:local_refinement}.
Finally, if all $b_I^T$ fulfill both lower- and upper-bounds, we can use these values to computer an assignment $\mu$ that on the one hand satisfies $\varphi$ (and thus the global constraints) and on the other hand satisfies all the local cardinality constraints (see Lemma~\ref{lem:realisationEquivalenceForSigma}).
If such $b_I^T$'s do not exist it is impossible to simultaneously satisfy the local and the global cardinality constraints.

\subsubsection{Extended Numerical Assignments}
Fix a formula $\varphi$ with $\ell$ free variables $X_1, \ldots, X_\ell$.
Let $G = (V, E)$ be a graph and let $\TT$ be its neighborhood diversity decomposition.
The \emph{extended numerical assignment} is a function $\sigma \colon 2^{\{X_1, \ldots, X_\ell\}} \times \TT \to \left[ |V(G)| \right]$.
We say that $\sigma$ is \emph{valid for $G$ and $\TT$} if \[\sum_{I \in 2^{\{X_1, \ldots, X_\ell\}}} \sigma(I, T) \le |T| \] for each type $T \in \TT$.
The crucial thing is that the extended numerical assignment plays the same role for the purpose of a design of the \XP algorithm as pre-evaluations for the \FPT algorithm presented in Section~\ref{sec:nd_fptalgo}.
We formalize this by showing that knowing $\sigma$ it is possible to decide whether $\varphi$ holds for $G$ or not.
Before we do so, we have to introduce one more formalism.
We say that an assignment $\mu \colon \left\{X_1, \ldots, X_\ell \right\} \to 2^V$ is a \emph{realisation} of a valid extended numerical assignment $\sigma$ if
\[
  \left| \left( \bigcap_{i \in I} \mu(X_i) \right) \cap T \right| = \sigma(I, T)
\]
for every $I \subseteq \left\{ X_1, \ldots, X_\ell \right\}$ and every type $T \in \TT$.
\begin{lemma}\label{lem:realisationEquivalenceForSigma}
Fix a formula $\varphi$ with $\ell$ free variables $X_1, \ldots, X_\ell$.
Let $G$ be a graph, let $\TT$ be its neighborhood diversity decomposition, and let $\sigma$ be a valid extended numerical assignment.
Then either each realization $\mu$ of $\sigma$ satisfies $\varphi$ or no realization of $\sigma$ satisfies $\varphi$.
\end{lemma}
\begin{proof}
Let assignments $\mu$ and $\mu'$ be realizations of $\sigma$.
We will show that $G, \mu \models \varphi$ if and only if $G, \mu' \models \varphi$.
This is not hard to see.
Note that the global cardinality constraints $R_1, \ldots, R_c$ contained in $\varphi$ depend solely on the number of vertices contained in some $X_i$'s and some types of $\TT$ which is, however, prescribed by $\sigma$ and thus both $\mu, \mu'$ have to agree on this.
Now, it is possible to evaluate the validity of all of the constraints $R_1, \ldots, R_c$ and (as in the case of pre-evaluations) replace them by the constants \texttt{true} or \texttt{false} in $\varphi$ yielding a simplified formula $\tilde\varphi$.
The lemma now follows from the fact that $\tilde\varphi$ is an \MSO formula and both $\mu$ and $\mu'$ yield the same labeled graph.
\end{proof}
By this we have shown that it makes sense to write $G, \sigma \models \varphi$.
Observe that there are at most $|V(G)|^{|\TT| \cdot 2^\ell}$ (valid) extended numerical assignments for $G$.

\subsubsection{Satisfying Local Cardinality Constraints}
Fix a valid extended numerical assignment $\sigma$ with $G, \sigma \models \varphi$.
Note that in such a case there is at least one realisation of~$\sigma$.
Now, we would like to resolve whether among all of the possible realizations of~$\sigma$ there is at least one realization that obeys the local cardinality constraints $(\alpha_1, \ldots, \alpha_\ell)$.

Fix a type $T \in \TT$ and define the function $s \colon \{ 1, \ldots, \ell \} \to \NN$ by
\[
  s(i) = \left| \bigcup_{S\in \TT : \{S, T\} \in E(T_G)} \left( S \cap \mu(X_i) \right) \right|.
\]
Recall that if $T$ is a clique in $G$, it has a loop in the corresponding type graph $T_G$.
Now, following Proposition~\ref{prop:solutionInType} we say that $(\alpha_1, \ldots, \alpha_\ell)$ are \emph{possibly satisfied by $\sigma$} if
\begin{itemize}
  \item
  $s(i) \in \alpha_i(v)$ for all $i \in [\ell]$ and $v \in T$, where $T$ is an independent set in $G$ and
  \item
  $\{ s(i)-1, s(i) \} \cap \alpha_i(v) \neq \emptyset$ for all $i \in [\ell]$ and $v \in T$, where $T$ is a clique in $G$.
\end{itemize}
Observe that if $T$ is an independent set in $G$ and $(\alpha_1, \ldots, \alpha_\ell)$ are possibly satisfied by $\sigma$, then every assignment $\mu$ realizing $\sigma$ fulfills $(\alpha_1, \ldots, \alpha_\ell)$ for every vertex in $T$.
This is, however, not true when $T$ is a clique in $G$.
In this case we let $t^+_i$ be the number of vertices in $T$ with $\{ s(i)-1, s(i) \} \cap \alpha_i(v) = \{ s(i)-1\}$, we let $t^-_i$ be the number of vertices in $T$ with $\{ s(i)-1, s(i) \} \cap \alpha_i(v) = \{ s(i) \}$, and we let $t^\pm_i$ be $|T| - (t^+_i + t^-_i)$.
Note that $t^+_i$ is the number of vertices in $T$ that must belong to $X_i$, $t^-_i$ is the number of vertices in $T$ that cannot belong to $X_i$, and $t^\pm_i$ is the number of vertices in $T$ that may or may not belong to $X_i$ (again this directly follows from Proposition~\ref{prop:solutionInType}).
Thus, we arrive at the following claim.

\begin{lemma}\label{lem:extendingPossibleSigma}
If $(\alpha_1, \ldots, \alpha_\ell)$ are possibly satisfied by $\sigma$ and for each type $T \in \TT$ forming a clique in $G$ we have
\[ t^+_i \le s(i) \le t^+_i + t^\pm_i \,, \]
then there is an assignment $\mu$ realizing $\sigma$ and fulfilling all of $(\alpha_1, \ldots, \alpha_\ell)$.
\qed
\end{lemma}

\begin{proof}[Proof of Theorem~\ref{thm:HYDEisXPwrtND}]
Let $G$ be a graph with $n = |V(G)|$ and neigborhood diversity $\nu$.
There are $n^{\nu \cdot 2^\ell}$ possible (valid) extended numerical assignments for $G$.
We loop through all of them and for each such $\sigma$ we
\begin{enumerate}
  \item check if $G, \sigma \models \varphi$,
  \item check if $\sigma$ possibly satisfies $(\alpha_1, \ldots, \alpha_\ell)$, and
  \item verify the conditions given in Lemma~\ref{lem:extendingPossibleSigma}.
\end{enumerate}
If all three of the above conditions are satisfied, we accept $\sigma$ and say Yes, otherwise we reject $\sigma$ and proceed to next $\sigma$.
If there is no $\sigma$ left, we say No.
It is not hard to see that the above procedure takes $O(n \cdot n^{\nu \cdot 2^\ell})$ time and can be simply extended so that it actually returns the sought assignment $\mu$ in the same time.
\end{proof}

\section{Theorem~\ref{thm:HYDEisXPwrtTW}: \XP algorithm for \MSOGL on bounded treewidth}\label{sec:tw_xpalgo}

We believe that the merit of Theorem~\ref{thm:HYDEisXPwrtTW} lies not only in being a very general tractability result, but also in showcasing a simplified way to prove a metatheorem extending \MSO.
Our main tool is the constraint satisfaction problem (CSP).
The key technical result of this section is Theorem~\ref{thm:mso_csp_extension}, which relates \MSO and CSP on graphs of bounded treewidth.
Let us shortly describe it.

Notice that in the \MSOGL model checking problem, we wish to find a satisfying assignment of some formula $\varphi$ which satisfies further constraints.
Simply put, Theorem~\ref{thm:mso_csp_extension} says that it is possible to restrict the set of satisfying assignments of a formula $\varphi \in$ \MSOD with CSP constraints under the condition that these additional constraints are structured along the tree decomposition of $G$.
This allows the proof of Theorem~\ref{thm:HYDEisXPwrtTW} to simply be a CSP formulation satisfying this property.

We believe that the key advantage of our approach, when compared with prior work, is that it is \emph{declarative}: it only states what a solution looks like, but does not describe how it is computed.
This makes the proof cleaner, and possible extensions easier.

We consider a natural optimization version of \MSOGL:

\probl{\textsc{Weighted \MSOGL}}
{An \MSOGL model checking instance, weights $\vc w^1, \dots, \vc w^\ell \in \ZZ^n$.}
{Find an assignment $X_1, \dots, X_\ell$ satisfying the \MSOGL model checking instance and minimizing $\sum_{j=1}^\ell \sum_{v \in X_j} w_v^j$.}
\subsection{CSP, MSO and treewidth}

\begin{definition}[CSP]
An instance $I=(V,\DD,\HH,\SSS)$ of CSP consists of
\begin{itemize}\parskip-3pt
\item a set of {\em variables} $z_v$, one for each $v\in V$; without loss of
generality we assume that $V = [|V|]$,
\item a set $\DD$ of finite {\em domains} $D_v\subseteq \ZZ$,
one for each $v\in V$,
\item a set of {\em hard constraints} $\HH \subseteq \{C_{U} \mid U \subseteq V \}$ where each
hard constraint $C_{U} \in \HH$ with a \emph{scope} $U=\{i_1,\dots,i_k\}$ and
$i_1 < \cdots < i_k$, is a $|U|$-ary relation
$C_U \subseteq D_{i_1}\times \cdots \times D_{i_k}$,
\item a set of {\em weighted soft constraints} $\SSS \subseteq \{w_U \mid
  U \subseteq V\}$ where each ${w_U \in \SSS}$ with a \emph{scope} $U=\{i_1,\dots,i_k\}$ and
${i_1 < \cdots < i_k}$ is a function $w_U:D_{i_1}\times \cdots\times D_{i_k}\rightarrow \RR$.
\end{itemize}
For a vector ${\vc z^{}=(z^{}_1, \ldots,z^{}_{|V|})}$ and a set $U=\{i_1,\dots,i_k\}\subseteq V$ with ${i_1 < \cdots < i_k}$, we define
the {\em projection of} $\vc z$ on $U$ as
$\vc z^{}|_U=(z^{}_{i_1}, \ldots, z^{}_{i_k})$.
A vector $\vc z\in \ZZ^{|V|}$ {\em satisfies the hard constraint} $C_U \in \HH$ if and only if $\vc z|_U \in C_U$.
We say that a vector
${\vc z^{\star}=(z^{\star}_1,\ldots,z^{\star}_{|V|})}$ is {\em a feasible assignment} for
$I$ if ${\vc z^{\star} \in D_1\times \cdots\times D_{|V|}}$ and
$\vc z^{\star}$ satisfies every hard constraint $C\in \HH$, and write $\Feas(I) = {\{\vc z^{\star} \mid \vc z^{\star}} \mbox{ is a feasible}\allowbreak \mbox{ assignment for } I\}$.
The \emph{weight} of $\vc z^{\star}$ is $w(\vc z^{\star})
= \sum_{w_U \in \SSS}w_U(\vc z^{\star}|_U)$.
To solve a CSP instance $I$ means to find a feasible assignment $\vc z$ which minimizes the weight $w(\vc z^*)$.

We denote by $D_I$ the maximum size of all domains,
that is, $D_I=\max_{u\in V}|D_u|$, and we omit the subscript $I$ if the instance is clear from the context.
We denote by $\|\DD\|$, $\|\HH\|$ and $\|\SSS\|$ the \emph{length} of $\DD$, $\HH$ and $\SSS$, respectively, and define it as $\|\DD\| = \sum_{v \in V} |D_v|$, $\|\HH\| = \sum_{C_U \in \HH} |C_U|$ and $\|\SSS\| = \sum_{w_U} |w_U|$; here $|w_U|$ denotes the size of the subset of $D_{i_1}\times \cdots\times D_{i_k}$ for which the function $w_U$ is nonzero.
\end{definition}

\begin{definition}[Constraint graph, Treewidth of CSP]
For a CSP instance $I=(V,\DD,\HH,\SSS)$ we define
the \emph{constraint graph} $G(I)$ of $I$ as $G=(V,E)$
where \[E= \left\{\{u,v\} \mid (\exists C_{U} \in \HH) \vee (\exists w_U \in
\SSS)\textrm{ s.t. } \{u,v\} \subseteq U, u \neq v\right\}.\]
The {\em treewidth of a CSP instance $I$}, $tw(I)$, is defined as
$\tw(G(I))$.
When we talk about $G(I)$ we use the terms ``variable'' and ``vertex'' interchangeably.
\end{definition}

Freuder~\cite{Freuder:90} proved that CSPs of bounded treewidth can be solved quickly.
We use a natural weighted version of this result.

\begin{proposition}[\cite{Freuder:90}]\label{prop:csp_tw}
  For a CSP instance $I$ of treewidth $\tau$ and maximum domain size $D$,
  a minimum weight solution can be found in time ${\Oh(D^\tau |V| +
  \|\HH\| + \|\SSS\|)}$.
\end{proposition}

Modeling after the terminology regarding extended formulations of polytopes~\cite{ConfortiCZ:2013}, we introduce the notion of a \emph{CSP extension}.

\begin{definition}[CSP extension]
Let $I = (V_I, \DD_I, \HH_I, \SSS_I)$ be a CSP instance.
We say that $J = (V_J, \DD_J, \HH_J, \SSS_J)$ is an \emph{extension of $I$} (or that \emph{$J$ extends $I$}) if $V_I \subseteq V_J$ and $\Feas(I) = \{\vc z^{\star}|_{V_I} \mid \vc z^{\star} \in \Feas(J)\}$.
\end{definition}

By Proposition~\ref{prop:csp_tw}, we can solve CSP instances of small treewidth efficiently.
Our motivation for introducing CSP extensions is that we are able to formulate a CSP instance $I$ expressing what we need, but having large treewidth.
However, if an extension $J$ of $I$ exists with small treewidth, solving $J$ instead suffices.

Let $\varphi$ be an \MSOD formula with $\ell$ free set variables and let $G$ a $\sigma_2$-structure with a universe of size $n$.
We say that a binary vector $\vc y \in \{0,1\}^{n \ell}$ \emph{satisfies $\varphi$} ($G, \vc y \models \varphi$) if it is the characteristic vector of a satisfying assignment $\mu$, that is, if $v \in \mu(X_i) \Leftrightarrow y_v^i = 1$ and $G, \mu \models \varphi$.
For a vector $\vc s$, let the \emph{support of $\vc s$} be $\textrm{supp}(\vc s) = \{i \mid s_i \neq 0\}$, that is, the set of its nonzero indices.
The following definition characterizes sets which are structured along a given tree decomposition of a graph.
The subsequent theorem then shows that, provided a CSP instance whose constraints are structured in this way along a tree decomposition of $G$, there exists a compact and tree-structured CSP extension.

\begin{definition}[Local scope property] \label{def:local_scope}
Let $\ell, m \in \NN$, $G$ be a $\sigma_2$-structure, $(T,\BB)$ be a nice tree decomposition of $G$, and $S$ be a set of vectors of elements indexed by $U := (V(G) \times [\ell]) \cup (V(T) \times [m])$.
We say that $S$ has the \emph{local scope property} if
\begin{align*}
\forall \vc s \in S~\exists a \in V(T): ~ \textrm{supp}(\vc s) \subseteq
	\big( & \{ (v,i) \mid v \in B(a), i \in [\ell]\} ~ \cup \\ & \{(b,j) \mid b \in N^{\downarrow}_T(a), j \in [m]\} \big).
\end{align*}
We extend the definition to a set $S$ containing not only vectors but also mappings indexed in the same way, where for $U' \subseteq U$ and a mapping $w_{U'}: \ZZ^{U'} \to \RR$, we define $\textrm{supp}(w_{U'}) = U'$.
\end{definition}

The heart of our proof of Theorem~\ref{thm:HYDEisXPwrtTW} is the following theorem.
\begin{theorem}\label{thm:mso_csp_extension}
Let $I=(V,\DD, \HH, \SSS)$ be a CSP instance,
$G$ a $\sigma_2$-structure, $\varphi$ an $\MSO_1$ formula with $\ell$ free variables, $(T, \BB)$ a nice tree decomposition of $G$ of width $\tau$, an integer $k \in \NN$, and a set of hard constraints $\HH'$ so that $V$ and $\HH$ satisfy
\[V = \{y_v^i \mid v \in V(G), i \in [\ell]\} \cup \{x_a^j \mid a \in V(T), j \in [k]\} \text{ and } \HH = \{\vc y \mid G, \vc y \models \varphi\} \cup \HH'.\]

If $\HH' \cup \SSS$ have the local scope property,
then there is a computable function $f$ and an algorithm computing in time $f(\|\varphi\|, \tau) \cdot |V| + \|\HH' + \SSS\|)$ a CSP instance $J = (V_J, \DD_J, \HH_J, \SSS_J)$ which extends $I$, and,
\begin{itemize}
\item $\tw(J) \leq f(||\varphi||, \tau) + 2k$,
\item $\|\HH_J\| + \|\SSS_J\| \leq f(||\varphi||, \tau) \cdot |V| + (\|\HH'\| + \|\SSS\|)$,
\item $D_J = D_I$.
\end{itemize}
\end{theorem}
Before giving the proof, we will show how Theorem~\ref{thm:mso_csp_extension} implies Theorem~\ref{thm:HYDEisXPwrtTW}.

\subsection{CSP instance construction}

\begin{proof}[Proof of Theorem~\ref{thm:HYDEisXPwrtTW}]
As before, we first note that there are at most $2^{||\varphi||}$ different pre-evaluations $\beta(\varphi)$ of $\varphi$, so we can try each and choose the result with minimum weight.
Let a pre-evaluation $\beta(\varphi)$ be fixed from now on.

Let $(T, \BB)$ be a nice tree decomposition of $G$.
We will now construct a CSP instance $I$ satisfying the conditions of Theorem~\ref{thm:mso_csp_extension}, which will give us its extension $J$ with properties suitable for applying Freuder's algorithm (Proposition~\ref{prop:csp_tw}).

Let $y_v^i$ be the variables as described above in Theorem~\ref{thm:mso_csp_extension}; we use the constraint $G, \vc y \models \beta(\varphi)$ to enforce that each feasible solution complies with the pre-evaluation $\beta(\varphi)$.
Now we will introduce additional CSP variables and constraints in two ways to assure that the local and global cardinality constraints are satisfied.
Observe that we introduce the additional CSP variables and constraints in such a way that they have the local scope property (Definition~\ref{def:local_scope}), that is, their scopes will always be limited to the neighborhood of some node $a \in V(T)$.

\subsubsection*{Global cardinality constraints.}
In addition to the original $\vc y$ variables, we introduce, for each node $a \in T$ and each $j \in [\ell]$, a variable $s_a^j$ with domain $[n]$.
We refer to the set of these variables as to $s$-variables.
The meaning of this variable is $s_a^j = |X_j \cap V(G_a)|$.
Thus, in the root node $r$, $s_r^j$ is exactly $|X_j|$.
To enforce the desired meaning of the variables $\vc s$, we add the following hard constraints:
\begin{align*}
s_a^j &= 0 & \forall \textrm{ leaves } a \\
s_a^j &= s_b^j + y_v^j & \forall a=b*(v) \\
s_a^j &= s_b^j & \forall a=b \dagger(v) \\
s_a^j &= s_b^j + s_{b'}^j - \sum_{v \in B(a)} y_v^j & \forall a = \Lambda(b,b')
\end{align*}
To enforce the cardinality constraints themselves, we add:
\begin{align*}
(s_r^1, \dots, s_r^\ell) &\in R & \forall R: \beta(R)=\texttt{true} \\
(s_r^1, \dots, s_r^\ell) &\in ([n]^\ell \setminus R) & \forall R: \beta(R)=\texttt{false}
\end{align*}

\subsubsection*{Local cardinality constraints.}
For every node $a \in V(T)$, every $j \in [\ell]$ and every vertex $v \in B(a)$, introduce a variable $\lambda_a^{vj}$ with domain $[n]$, with the meaning $\lambda_a^{vj} = |N_{G_a}(v) \cap X_j|$.
We refer to these as to $\lambda$-variables.
Their meaning is enforced by setting:
\begin{align*}
\lambda_a^{vj} &= \sum_{\substack{u \in B(a): u \neq v \\ uv \in E}} y_u^j & \forall a=b*(v) \\
\lambda_a^{uj} &= \lambda_b^{uj} + y_v^j & \forall a=b*(v), u \in B(a), uv \in E \\
\lambda_a^{vj} &= \lambda_b^{vj} & \forall a = b \dagger(u), u \neq v, v \in B(a) \\
\lambda_a^{vj} &= \lambda_b^{vj} + \lambda_{b'}^{vj} - \sum_{\substack{u \in B(a): u \neq v \\ uv \in E}} y_u^j & \forall a = \Lambda(b,b')
\end{align*}
Now the local constraints themselves are enforced by setting:
\begin{align*}
\lambda_{\textrm{top}(v)}^{vj} &\in \alpha_j(v) & \forall v \in V, \quad j\in [\ell] \enspace .
\end{align*}

\subsubsection*{Objective function.} In order to express the objective function we add soft constraints $\SSS =\{w_{\{y_v^j\}} \mid v \in V, j \in [\ell]\}$ where $w_{\{y_v^j\}}(y_v^j) = w_v^j$ if $y_v^j = 1$ and is $0$ otherwise.

\medskip

In order to apply Theorem~\ref{thm:mso_csp_extension}, let us determine the parameters of the CSP instance $I$ we have constructed, namely the number of extra variables per node $k$, the maximum domain size $D_I$, and the lengths of the additional constraints $\|\HH'\|$ and $\|\SSS\|$.

We have introduced $\ell$ $s$-variables per node, and $\ell \tau$ $\lambda$-variables per node.
Thus, $k = (\tau +1) \ell$.
Since each $s$- and $\lambda$-variable corresponds to a size of some vertex subset, its value is upper bounded by $n$, and thus $D_I = n$.
Let \[N = \sum_{j=1}^{c} |R_j^G| + \sum_{j=1}^{\ell} \sum_{v \in V(G)}|\alpha_j(v)|\] be the input length of the global and local cardinality constraints.

Let $\HH'$ be the set of all CSP constraints defined above to enforce the global and local constraints of the \MSOGL model checking instance.
Then, Theorem~\ref{thm:mso_csp_extension} implies that we can compute a CSP instance $J$ which is an extension of $I$ and has
\begin{itemize}
\item $\tw(J) \leq f(||\varphi||, \tau) + 2k \leq f(||\varphi||,\tau) + (\tau+1)\ell \leq f'(||\varphi||, \tau)$ for some computable function $f'$,
\item $\|\HH_J\| + \|\SSS_J\| \leq f(||\varphi||, \tau) \cdot |V| + (\|\HH'\| + \|\SSS\|) \leq f(||\varphi||, \tau) \cdot n + N$, and
\item $D_J = D_I = n$.
\end{itemize}

Finally, applying Proposition~\ref{prop:csp_tw} to $J$ solves it in time $n^{f'(||\varphi||, \tau)} + N$, finishing the proof of Theorem~\ref{thm:HYDEisXPwrtTW}.
\end{proof}

\subsubsection*{Conditional cardinality constraints}
As Szeider~\cite{Szeider:11} points out, it is easy to extend his \XP result for \MSOL in such a way that the local cardinality constraint $|X(v)| \in \alpha(v)$ is conditioned on the fact that $v \in X$.
Observe that our approach can be extended in such a way as well by enforcing
\begin{align*}
y_v^j = 1 \implies \lambda_{\textrm{top}(v)}^{vj} &\in \alpha_j(v) & \forall v \in V, \quad j\in [\ell] \enspace .
\end{align*}
(Formally, the above is a binary relation of the variables $y_v^j$ and $\lambda_{\textrm{top}(v)}^{vj}$ defined as $R = \{(0, i) \mid i \in [n]\} \cup \{(1,i) \mid i \in \alpha_j(v)\}$.)
Moreover, in our setting with multiple set variables, we can even condition on an arbitrary predicate $\psi(v, X_1, \dots, X_m)$ describing how vertex $v$ relates to the set variables.

\subsection{Applications (Corollary~\ref{thm:HYDE_applications}) } \label{subsec:hyde_tw_apps}

Let us briefly sketch some consequences of Theorem~\ref{thm:HYDEisXPwrtTW}.
We focus on showing how to encode various \Wh{1} (w.r.t. treewidth) problems using the notions we have provided.
The parameterized complexity statements which follow are not very surprising and in many cases were known.
Still, we believe that our approach captures and summarizes them nicely.

\subsubsection*{Local constraints}
While introducing \MSOL, Szeider~\cite{Szeider:11} points out that the problems \textsc{General Factor}, \textsc{Equitable $k$-Coloring} and \textsc{Minimum Maximum Outdegree} are expressible in \MSOL.
Let us now observe that using the extension to conditional local constraints, we can also express the problems \textsc{Capacitated Dominating Set}, \textsc{Capacitated Vertex Cover}, \textsc{Vector Dominating Set} and \textsc{Generalized Domination}.

Take for example the \textsc{Capacitated Dominating Set} problem.
There, we are given a graph $G=(V,E)$ together with a capacity function $c: V \to \NN$, and our goal is to find a subset $D \subseteq V$ and a mapping $f: V \setminus D \to D$ such that for each $v \in D$, $|f^{-1}(v)| \leq c(v)$. Essentially, $f(w)=v$ means that the vertex $w$ is dominated by the vertex $v$, and the condition $|f^{-1}(v)| \leq c(v)$ ensures that the mapping $f$ respects the capacities.
Recall that \MSOL here generalizes \MSOD, so we view $G$ as the $\sigma_2$-structure whose universe $V_I$ contains both vertices and edges of the graph it represents.
Then, we let $\varphi(D,F)$ and $\alpha$ be a formula and local cardinality constraints, respectively, enforcing that:
\begin{itemize}
\item $D \subseteq V$,
\item $F \subseteq E$,
\item each $v \in V$ is either in $D$ or has a neighbor in $F$,
\item each $e \in F$ has a neighbor in $D$, and,
\item if $v \in D$, then $|N(v) \cap F| \in \alpha_F(v) = [0, c(v)]$.
\end{itemize}
Then $D$ encodes a dominating set and $F$ can be used to construct the mapping $f$ since for each edge $e \in F$, at most one of its endpoints is not in $D$, and each $v \in D$ sees at most $c(v)$ edges from $F$.

Let us define the remaining problems; their \MSOL formulations are analogous to the one above.
The \textsc{Vector Dominating Set} problem is similar to \textsc{Capacitated Dominating Set}, except now each vertex $v$ has a \emph{demand} $d(v)$ and if $v \not\in D$, then it must have at least $d(v)$ neighbors in $D$.
In \textsc{Generalized Domination}, we are given two sets $\sigma, \rho \subseteq \NN$ and for each vertex $v$ in $D$ or in $V \setminus D$, it must hold that $|N(v) \cap D| \in \sigma$ or $|N(v) \cap D| \in \rho$, respectively.
Finally, the \textsc{Capacitated Vertex Cover} problem is the following.
We are given a capacitated graph, and the task is to find a vertex cover $C \subseteq V$ and an assignment $f: E \to C$ such that for each $v$, $|f^{-1}(v)| \leq c(v)$.

\subsubsection*{Global constraints.}
Ganian and Obdržálek~\cite{GO:13} introduce \MSOGLin and show that also this logic expresses \textsc{Equitable $k$-Coloring}, and moreover the \textsc{Equitable connected $k$-Partition} problem.
Interestingly, they also discuss the complexity of the \textsc{$k$-Balanced Partitioning} problem, where the goal is to find an equitable (all parts of size differing by at most one) $k$-partition and, moreover, minimize the number of edges between partites.
They provide an \FPT algorithm for graphs of bounded vertex cover, but are unable to express the problem in \MSOGLin, and thus pose as an open problem the task of finding a more expressive formalism which would capture this problem.
They also state that no parameterized algorithm exists for graphs of bounded treewidth, but that is no longer true due to the results of van Bevern et al.~\cite{BevernFSS:2015}.
On the other hand, the question of capturing \textsc{$k$-Balanced Partitioning} by some \MSO extension stands.
Here we show that it can be expressed as an instance of weighted \MSOGLin model checking over $\sigma_2$-structures (thus this is not applicable to graphs of bounded neighborhood diversity where we only have algorithms for $\sigma_1$-structures).

Let $\varphi$ be an \MSOG formula with $k$ free vertex set variables $X_1, \dots, X_k$ and one free edge set variable $Y$.
We use $\varphi$ to express that $X_1, \dots, X_k$ is an equitable $k$-partition; this is easily done using the global constraints.
Furthermore, we enforce that $Y$ is the set of edges with one endpoint in $X_i$ and another in $X_j$ for any $i \neq j$.
For a satisfying assignment $X_1, \dots, X_k, Y$, let $(\vc x, \vc y) \in \{0,1\}^{k|V|} \times \{0,1\}^{|E|}$ be its characteristic vector.
To minimize the number of edges between partites, it suffices to minimize $\vc y$.
This also clearly extends to the case studied in the literature when edges are assigned weights.

\medskip

Fellows et al.~\cite{FellowsFHV:2011} study the \textsc{Graph Motif} problem from the perspective of parameterized complexity, especially on graphs with bounded widths.
In \textsc{Graph Motif}, we are given a vertex-colored graph $G$ with $\chi$ colors and a multiset of colors $M$, and the task is to find a \emph{motif}, that is, a connected subset of vertices $S \subseteq V$ such that the multiset of colors of $S$ is exactly $M$.
This problem is most naturally expressed when the vertex-colored graph $G$ is encoded as a $\sigma$-structure with $\sigma = (V \cup E, \emptyset, \{L_V, L_E, L_1, \dots, L_\chi\})$ (i.e., $\sigma$ is $\sigma_2$ extended by unary predicates $L_1, \dots, L_\chi$).
Since we have not explicitly phrased our results for such structures, \textsc{Graph Motif} does not directly fit any of our notions.
However, the extension of our results (specifically, Theorem~\ref{thm:HYDEisXPwrtTW}) to such structures is straightforward since it is known that Courcelle's theorem can be extended and then for example global constraints over the set $S \cap L_i$ are treated like global constraints over any other set (there is no change to the CSP constraints required in the proof of Theorem~\ref{thm:HYDEisXPwrtTW}).

Then, let us consider the number of colors $\chi$ a parameter and introduce additional unary relations (labels) $L_1, \dots, L_\chi$.
It is easy to see that \textsc{Graph Motif} is encoded by the following \MSOGLin formula $\varphi(S)$:
\[\varphi(S) \equiv \text{connected}(S) \wedge \bigwedge_{i=1}^\chi [|S \cap L_i| = \text{mult}(i,M)],\]
where $\text{connected}(S)$ is a formula which holds if $S$ is connected, and $\text{mult}(i,M) \in \NN$ is the multiplicity of color $i$ in the motif $M$.

\subsection{Proof of Theorem~\ref{thm:mso_csp_extension}}
Fix objects and quantities as in the statement of Theorem~\ref{thm:mso_csp_extension}, that is, $I=(V,\DD, \HH, \SSS)$ is a CSP instance, $G$ is a $\sigma_2$-structure, $n$ is the size of the universe of $G$, $\varphi$ is an \MSOD formula with $\ell$ free variables, $(T, \BB)$ is a nice tree decomposition of $G$ of width $\tau$, $k \in \NN$ is such that $V$ has $\ell \cdot |V(G)|$ variables $y_v^i$ and $k \cdot |V(T)|$ variables $x_a^j$, $\HH = \{\vc y \mid G,\vc y \models \varphi\} \cup \HH'$ and $\HH' \cup \SSS$ has the local scope property.

We give a brief outline of the proof of Theorem~\ref{thm:mso_csp_extension}, which proceeds in three stages:
\begin{enumerate}
\item Using a recent result of Kolman, Koutecký and Tiwary~\cite{KolmanKT:2015} we construct a linear program (LP) whose constraint matrix has bounded treewidth (to be defined) and whose integer solutions correspond to feasible assignments of $\varphi$.
\item We view this LP as an integer linear program (ILP) and construct an equivalent CSP instance $J'$ of bounded treewidth.
Thus, $J'$ is an extension of $I'=([n \cdot \ell], \DD_{I'}, \HH_{I'}, \emptyset)$ where $\DD_{I'} = \{D_i \mid D_i = \{0,1\}, \, i \in [n \cdot \ell]\}$ and $\HH_{I'} = \{\vc y \mid G, \vc y \models \varphi\}$.
\item We show that if $\HH'$ and $\SSS$ have the local scope property, it is possible to add new constraints derived from $\HH'$ and $\SSS$ to the instance $J'$ which results in instance $J$ (with $\Feas(J) \subseteq \Feas(J')$), such that $J$ is an extension of $I$.
\end{enumerate}

\medskip

\subsubsection*{Stage 1: LP}
We begin by extending our notion of treewidth to matrices.
\begin{definition}[Treewidth of a matrix]
Given a matrix $\vc A \in \ZZ^{m \times n}$, we define its
\emph{Gaifman graph} $G = G(\vc A)$ as follows. Let $V(G) = [n]$.
Let $\{i,j\} \in E(G)$ if and only if there is an
$r \in [n]$ with $\vc A[r, i] \neq 0$ and $\vc A[r, j] \neq 0$.
The \emph{(primal or Gaifman) treewidth of a matrix $\vc A$} is then $\tw(\vc A) := \tw(G(\vc A))$.
\end{definition}

Let $P_\varphi(G) = \textrm{conv} \{\vc y \mid G, \vc y \models \varphi\}$ be the polytope of satisfying assignments of $\varphi$ on $G$, also called the \emph{\MSO polytope}; $\textrm{conv}(X)$ denotes the convex hull of a set $X$.
A result of Kolman et al.~\cite{KolmanKT:2015} shows that there exists a polytope closely related to $P_\varphi(G)$ with many useful properties:

\begin{proposition}[{\cite[Theorem 4]{KolmanKT:2015}}]\label{prop:courcelle_lp}
  Let $G=(V,E)$ be a $\sigma_2$-structure, let $n$ be the size of the universe of $G$,
$(T, \BB)$ be a nice tree decomposition of $G$ of width $\tau$, and $\varphi$ be an \MSOD formula with $\ell$ free variables.

  Then there exists an LP $\vc A \vc y + \vc D \vc z + \vc C \vc w = \vc d, \,\vc z,\vc w \geq 0$, a set $\CC$,
  a function $\eta: \CC\times V(T)\times V \times [\ell] \rightarrow \{0,1\}$
  and a tree decomposition $(T, \BB^*)$
  of the Gaifman graph $G(\vc A~\vc D~\vc C)$
  such that the following claims hold:
      \begin{enumerate}
  \item The polytope $P = \{(\vc y,\vc z,\vc w) \mid \vc A \vc y + \vc D \vc z + \vc C \vc w = \vc d, \, \vc z,\vc w \geq
    \mathbf{0}\}$ is a $0/1$-polytope and $P_\varphi(G) = \{\vc y \mid \exists \vc z, \vc w: (\vc y,\vc z,\vc w) \in P\}$.
  \item \label{thm:variables}
    For any integer point $(\vc y,\vc z,\vc w) \in P$,
    for any $t \in \CC$, $b \in V(T)$, $v \in B(b)$
    and $i \in [\ell]$, equalities $z_b^t = 1$ and
    $\eta(t,b,v,i) = 1$ imply that $y_v^i = 1$.
  \item \label{thm:tw_structure}
    \begin{enumerate}
      \item The treewidth of $(T, \BB^*)$ is $\Oh(|\CC|^3)$,
      \item for every node $b \in V(T)$, $\bigcup_{t\in \CC}\{z_b^t\} \subseteq B^*(b)$.
     \end{enumerate}
  \item There is a computable function $f$ such that $|\CC| \leq f(\tau, ||\varphi||)$ and $\vc A, \vc D, \vc C, \vc d, \eta$ can be computed in time $\Oh(|\CC|^3 \cdot n)$.
  \end{enumerate}
\end{proposition}

\subsubsection*{Stage 2: Viewing ILP as CSP}
By $\vc a_j$ we denote the $j$-th row of a matrix $\vc A$.
Let us connect ILP and CSP:

\begin{definition}
A \textit{CSP instance $I$ is equivalent to an ILP $\vc A \vc x \leq \vc b\, \vc x \in \ZZ^n$} if \[\{\vc x \in \ZZ^n \mid \vc A \vc x \leq \vc b\} = \Feas(I).\]
\end{definition}

\begin{proposition}\label{prop:ilp_csp}
Let $\vc A \vc x \leq \vc b,\, \vc x \in \ZZ^n$ be an ILP with $\vc A \in \ZZ^{m \times n}$, $\tau=\max_{j=1}^m|\textrm{supp}(\vc a_j)|$ and $D=\max_{i=1}^n |u_i - \ell_i|$, where $u_i$ and $\ell_i$ are an upper and lower bound on $x_i$, that is, $\vc u \leq \vc x \leq \vcl$ holds. Then, an equivalent CSP instance $I$ can be constructed in time $\Oh(D^\tau \cdot mn)$, and $G(\vc A) = G(I)$.
\end{proposition}
  \begin{proof}
  Let $V = \{x_1, \dots, x_n\}$. For every $i \in [n]$, let $D_i = [\ell_i, u_i]$ and $\DD = \{D_i \mid i \in [n]\}$. Observe that $\max_i |D_i| = \|\mathbf{u} - \vcl\|_{\infty} =
  D$. Regarding hard constraints $\HH$, observe that every row $\vc a_j$
  of $\vc A$ contains at most $\tau$ non-zeros. Let $U_j = \textrm{supp}(\vc a_j) = \{i_1, \dots,
  i_k\}$, where $k \leq \tau$, and let $x_c = 0$ for all $c \not\in U_j$. Let
  $C_{U_j}$ be the set of assignments from $D_{i_1} \times \cdots
  \times D_{i_k}$ to $x_{i_1}, \ldots, x_{i_k}$ that satisfy
  $\vc a_j \mathbf{x} \leq b_j$; obviously $|C_{U_j}|
  \leq D^k$ and it can be constructed in time $\Oh(D^k)$. Then, $\HH =
  \{C_{U_j} \mid j = 1, \dots, m\}$.
  It is easy to verify that the feasible assignments of $I$ correspond
  to integer solutions of $\vc A \mathbf{x} \leq \mathbf{b}$ and that $G(\vc A) = G(I)$.
\end{proof}

\begin{proof}[Proof of Theorem~\ref{thm:mso_csp_extension}]
We apply Proposition~\ref{prop:courcelle_lp} to obtain an ILP $\vc A \vc y + \vc D \vc z  + \vc C \vc w= \vc d$, and use Proposition~\ref{prop:ilp_csp} to get an equivalent CSP instance $J'$.
Recall that $(T, \BB)$ is a nice tree decomposition of $G$ and $(T, \BB^*)$ is a tree decomposition of $G(\vc A~\vc D~\vc C)$ (and thus $G(J')$) as described in Proposition~\ref{prop:courcelle_lp}, part~\eqref{thm:tw_structure}.
Let $I'$ be a CSP instance over variables $\vc y$ with $\HH = \{\vc y \mid G, \vc y \models \varphi\}$.
Clearly, $J'$ is an extension of $I'$ by the fact that the polytope $P$ is an extension of $P_\varphi(G)$.

Now we will add the variables $\vc x$ to $J'$ and add constraints in such a way that the resulting instance $J$ will be an extension of $I$, and that it satisfies the claim of Theorem~\ref{thm:mso_csp_extension}.

\subsubsection*{Stage 3: Adding variables and constraints}
We introduce auxiliary binary variables $f_v^{i,a}$ for each $a \in V(T)$, $i \in [\ell]$ and $v \in B(a)$, and we let $f_v^{i,a} = \sum_{t \in \CC} z_a^t \cdot \eta(t,a,v,i)$.
For any subset $U$ of variables of $I$, let $U_a$ be the set $U$ where each variable $y_v^i$ is replaced by $f_v^{i,a}$.
Then, for every constraint $C_U \in \HH'$ and $w_U \in \SSS$ let $a \in V(T)$ be the node in the definition of the local scope property which satisfies
\begin{equation*}
\textrm{supp}(\vc s) \subseteq \big( \{ (v,i) \mid v \in B(a), i \in [\ell]\} \cup  \{(b,j) \mid b \in N^{\downarrow}_T(a), j \in [m]\} \big),
\end{equation*}
and add to $J'$ a constraint obtained by replacing the scope $U$ with $U_a$.
Denote the resulting instance $J$.

By property~\eqref{thm:variables} of Proposition~\ref{prop:courcelle_lp}, $f_v^{i,a} = y_v^i$ for each $a \in V(T)$ with $v \in B(a)$.
Thus, for any constraint $C_U$ or $w_U$, replacing its scope with $U_a$ does not change the set of feasible assignments, and $J$ is an extension of $I$.

By the equivalence of $\vc A \vc y + \vc D \vc z  + \vc C \vc w= \vc d$ and $J'$, we have that $\|\HH_{J'}\| + \|\SSS_{J'}\| \leq f(||\varphi||, \tau) \cdot n$, and thus $\|\HH_{J}\| + \|\SSS_{J}\| \leq f(||\varphi||, \tau) \cdot n + \|\HH'\| + \|\SSS\|$.
The variables contained in $J$ and not $I$ are the $\vc z, \vc w$ and $\vc f$ variables.
Since they are all binary we also have that $D_J = D_I$.
It remains to show that $\tw(J) \leq f(||\varphi||, \tau) + 2k$.
A lemma will help us see that:
\begin{lemma}\label{lem:mod_td}
Let $T=(I,F)$ be a rooted binary tree, let $(T, \BB)$ be a tree decomposition of a graph $G=(V,E)$ of width $\kappa$, and let $H=(V \cup W,E \cup Y)$ be a supergraph of $G$ such that:
\begin{itemize}
\item $W = \bigcup_{a \in I} W_a$, $Y = \bigcup_{ab \in F} Y_{ab}$, all $W_a$ are pair-wise disjoint, and all $Y_{ab}$ are pair-wise disjoint,
\item $|W_a| \leq \kappa'$ for all $a \in I$,
\item if $a \in I$ has only child $b$, then $\bigcup Y_{ab} \subseteq (B(a) \cup W_a \cup W_b)$, and,
\item if $a$ has two children $b, b'$, then $\bigcup (Y_{ab} \cup Y_{ab'}) \subseteq (B(a) \cup W_a \cup W_b \cup W_{b'})$.
\end{itemize}
Then there is a tree decomposition $(T, \BB')$ of $H$ of width at most $\kappa+2\kappa'$.
\end{lemma}

\begin{figure}[ht]
  \begin{center}
\usetikzlibrary{positioning,calc,fit}

\begin{tikzpicture}
  \tikzstyle{bag}=[circle, draw, minimum width=.8cm, line width=2pt, inner sep=2pt]
  \tikzstyle{vertex}=[circle, draw, fill, minimum width=1pt, inner sep=3pt]
  \tikzstyle{dummyvertex}=[circle, minimum width=2pt]
  \tikzstyle{decompositionEdge}=[line width=2pt]
  \tikzstyle{edge}=[thick]
  \tikzstyle{treeEdge}=[line width=3pt, gray]

\begin{scope}[local bounding box=levy, node distance=.8cm]
  \node[bag] (Xa) {$W_a$};
  \node[bag] at ($(Xa) - (0,2cm)$) (Xb) {$W_b$};

  \node[vertex, label={0:$a$}] (a) at ($(Xa) + (2cm,0)$) {};
  \node[vertex, label={0:$b$}] (b) at ($(a) - (0,2cm)$) {};

  \draw[decompositionEdge] (Xa) -- ($(Xb.north) - (0,1pt)$);
  \draw[treeEdge] (a) -- (b);

  \draw[edge] (b) -- (Xb);
  \draw[edge] (b) -- (Xb.south east);
  \draw[edge] (b) -- ($(Xa.south east)!.25!(Xa.east)$);
  \draw[edge] (b) -- ($(Xa.south east)!.5!(Xa.east)$);
  \draw[edge] (b) -- ($(Xb.south east)!.75!(Xb.east)$);
  \draw[edge] (b) -- (Xb.north east);
  \draw[edge] (b) -- ($(Xb.north east)!.25!(Xb.east)$);
  \draw[edge] (b) -- ($(Xb.north east)!.5!(Xb.east)$);
  \draw[edge] (b) -- ($(Xb.north east)!.75!(Xb.east)$);

  \draw[edge] (a) -- (Xa.north east);
  \draw[edge] (a) -- ($(Xa.north east)!.25!(Xa.east)$);
  \draw[edge] (a) -- ($(Xa.north east)!.5!(Xa.east)$);
  \draw[edge] (a) -- ($(Xa.north east)!.75!(Xa.east)$);
  \draw[edge] (a) -- (Xa);

\end{scope}

\begin{scope}[shift={($(levy.north east) + (1cm, -1cm)$)}, node distance=1.2cm, local bounding box=pravy]
  \node[bag] (A) {$W_a$};
  \node[bag, above right of=A] (B) {$W_b$};

  \node[vertex, right of=A] (va) {};
  \node[vertex, right of=B] (vb) {};
  \node[vertex, below right of=vb] (vc) {};
  \node[vertex, below right of=vc] (vd) {};
  \node[vertex, below right of=vd] (ve) {};
  \node[vertex, below left of=vd] (vf) {};

  \node[bag, right of=vc] (C) {$W_c$};
  \node[bag, right of=vd] (D) {$W_d$};
  \node[bag, right of=ve] (E) {$W_e$};
  \node[bag, left of=vf] (F) {$W_f$};

  \draw[treeEdge] (va) -- (vb) -- (vc) -- (vd) -- (ve);
  \draw[treeEdge] (vd) -- (vf);

  \draw[decompositionEdge] (A) -- (B);
  \draw[decompositionEdge] (C) -- (D) -- (E);

  \draw[decompositionEdge] (B) to[out=30] (C);
  \draw[decompositionEdge] (D) to[out=150, in=60] (F);

  \draw[edge] (A) -- (va);
  \draw[edge] (B) -- (vb);
  \draw[edge] (C) -- (vc);
  \draw[edge] (D) -- (vd);
  \draw[edge] (E) -- (ve);
  \draw[edge] (F) -- (vf);
\end{scope}

\node[fit=(A) (B) (C) (D) (E) (F),label={270:(b)}] {};
\node[fit=(Xa) (Xb) (a) (b),label={270:(a)}] {};

\end{tikzpicture}
  \end{center}
  \caption{The situation of Lemma~\ref{lem:mod_td}. Part (a) depicts a single edge $ab \in F$ and the requirement that edges from vertices in $W_a$ only connect to vertices in $W_a$, $W_b$, $B(a)$ or $B(b)$. Part (b) depicts how $W$ and $Y$ relate to the whole of $T$. Black points correspond to $I$, grey edges to $F$, $W_a, \dots, W_f$ are self-explanatory, and all remaining edges correspond to sets of edges $Y_{ij}$.}
\end{figure}

\begin{proof}
Let $\BB'$ be obtained from $\BB$ by, for every edge $ab \in F$, adding $W_a$ to the bags $B(a)$ and $B(b)$.
We will verify that $(T, \BB')$ is a tree decomposition of $H$ of width at most $\kappa + 2\kappa'$; we shall denote by $B'$ the bags of $(T, \BB')$.
The conditions of a tree decomposition obviously hold for all vertices and edges of $G$, so we only check it for new vertices and edges.

\textbf{Edge condition.} Let $uv \in Y_{ab}$ be an edge in $H \setminus G$ with $u \in W_a$. Either $v \in B(a)$ and then $\{u,v\} \subseteq B(a) \cup W_a \subseteq B'(a)$, or $v \in W_b$ and then $\{u, v\} \subseteq B(b) \cup W_a \subseteq B'(b)$.

\textbf{Connectedness condition.} Let $v \in W_a$ and let $a$ have children $b,b'$, with possibly $b=b'$. Notice that $v$ does not appear in the bag of any node above $a$ and any node below $b$ and $b'$. Since we have added $W_a$ to all of $a, b$ and $b'$, the connectedness condition holds.

We have added to each node $b$ (except the root) two sets $W_a$, $W_b$ where $a$ is the parent of $b$, and because $|W_a| + |W_b| \leq 2\kappa'$, $\tw((T, \BB')) \leq \kappa+2\kappa'$.
\end{proof}

Let us consider how the constraint graph $G(J)$ relates to $G(J')$.
Since $J$ is obtained by adding new variables and constraints, this corresponds to $G(J)$ being a supergraph of $G(J')$.
The vertices $W=V(G(J)) \setminus V(G(J'))$ can be partitioned into sets $W_a$ for every node $a \in V(T)$, and $|W_a| \leq k$.
Moreover, the new edges $Y= E(G(J)) \setminus E(G(J'))$ can also be partitioned into sets $Y_{ab}$ for each $ab \in E(T)$, such that for each $uv \in Y_{ab}$ we have $\{u,v\} \subseteq (W_a \cup W_b)$, because, for each node $a \in V(T)$, the new constraints only contain variables associated with node $a$ and its neighbors.
The tree decomposition $(T, \BB^*)$ of $G(J')$ is such that we are precisely in the situation of Lemma~\ref{lem:mod_td} with $G := G(J')$, $H := G(J)$, $\kappa:=\tw(J') = f'(||\varphi||, \tau)$ and $\kappa':=k + \ell \tau$, which then implies that $G(J)$ has a tree decomposition $(T, \BB')$ of width $f'(||\varphi||, \tau)) + 2k + \ell \tau \leq f(||\varphi||, \tau)) + 2k$.
\end{proof}

\section{Conclusions}
\subsubsection*{Limits of \MSO extensions, other logics, and metatheorems.} We have defined extensions of \MSO and extended positive and negative results for them.
There is still some unexplored space in \MSO extensions: Szeider~\cite{Szeider:11} shows that \MSOL where some of the sets of local cardinality constraints are quantified is \NP-hard already on graphs of treewidth 2.
We are not aware of a comparable result for \MSOG, and no results of this kind are known for graphs of bounded neighborhood diversity.
Also, we have not explored other logics, as for example the modal logic considered by Pilipczuk~\cite{Pilipczuk:11}.
Also, one merit of algorithmic metatheorems is in generalizing existing results. However, many problems~\cite{FialaGKKK:16,Ganian:12} are \FPT on bounded neighborhood diversity which are not expressible in any of the studied logics. So we ask for a metatheorem generalizing as many such positive results as possible.

\subsubsection*{Complementary Parameters and Problems.}
Unlike for treewidth, taking the complement of a graph preserves its neighborhood diversity.
Thus our results apply also in the complementary setting, where, given a graph $G$ and a parameter $p(G)$, we are interested in the complexity (with respect to $p(G)$) of deciding a problem $P$ on the \emph{complement} of $G$.
While the complexity stays the same when parameterizing by neighborhood diversity, it is unclear for sparse graph parameters such as treewidth.
It was shown very recently~\cite{dvorakKM16} that the {\sc Hamiltonian Path} problem admits an \FPT algorithm with respect to the treewidth of the complement of the graph.
This suggest that at least sometimes this is the case and some extension of Courcelle's theorem deciding properties of the complement may hold.

\subsection*{Acknowledgements.}
This research was supported by the project 338216 of GA UK and the grant SVV--2017--260452.
D.~Knop acknowledges support by the OP VVV MEYS funded project CZ.02.1.01/0.0/0.0/16\_019/0000765 ``Research Center for Informatics''.
M.~Koutecký was partially supported by a postdoctoral fellowship at the Technion funded by the Israel Science Foundation grant 308/18, by Charles University project UNCE/SCI/004, and by the project 17-09142S of GA ČR.
T.~Masařík was supported by the project GA17-09142S of GA \v{C}R.

\bibliographystyle{plainurl}
\bibliography{metathm}

\end{document}